\documentclass[letterpaper, 10 pt, conference]{ieeeconf}%

\IEEEoverridecommandlockouts
\usepackage{cite}
\usepackage{amsmath,amssymb,amsfonts}
\usepackage{algorithmic}
\usepackage{graphicx}
\usepackage{textcomp}
\usepackage{xcolor}
\usepackage{cuted}
\usepackage{lipsum}
\usepackage{mathtools}
\usepackage{stfloats}
\usepackage{bbm}
\usepackage{multirow}

\newcommand{\R}{\mathbb{R}}
\newcommand{\N}{\mathbb{N}}
\newcommand{\T}{\top}
\newcommand{\I}{\mathbf{I}}
\newcommand{\0}{\mathbf{0}}
\newcommand{\E}{\mathcal{E}}

\newcommand{\diag}{\text{diag}}
\newcommand{\tsup}[1]{\textsuperscript{#1}}
\newcommand{\mb}[1]{\mathbf{#1}}

\newcommand{\bm}[1]{\begin{bmatrix}#1\end{bmatrix}}
\def\BibTeX{{\rm B\kern-.05em{\sc i\kern-.025em b}\kern-.08em
    T\kern-.1667em\lower.7ex\hbox{E}\kern-.125emX}}

\newtheorem{assumption}{\textbf{Assumption}}
\newtheorem{theorem}{\textbf{Theorem}}

\newtheorem{lemma}{\textbf{Lemma}}
\newtheorem{proposition}{\textbf{Proposition}}
\newtheorem{remark}{\textbf{Remark}}
\newtheorem{definition}{\textbf{Definition}}

\usepackage{hyperref}
    
\begin{document}

\title{
{\LARGE \textbf{
{Dissipativity-Based Distributed Control and Communication Topology Co-Design for Nonlinear DC Microgrids}
}}}



\author{Mohammad Javad Najafirad and Shirantha Welikala 
\thanks{The authors are with the Department of Electrical and Computer Engineering, School of Engineering and Science, Stevens Institute of Technology, Hoboken, NJ 07030, \texttt{{\small \{mnajafir,swelikal\}@stevens.edu}}.}}

\maketitle


\begin{abstract}
This paper presents a dissipativity-based distributed droop-free control and communication topology co-design framework for voltage regulation and current sharing in DC microgrids (MGs), where constant-power loads (CPLs) and voltage-source converter (VSC) input saturation introduce significant nonlinearities. In particular, CPLs introduce an inherently destabilizing nonlinearity, while VSC input saturation imposes hard amplitude constraints on applicable control input at each distributed generator (DG), collectively making the DC MG control system design extremely challenging. 
To this end, the DC MG is modeled as a networked system of DGs, transmission lines, and loads coupled through a static interconnection matrix. Each DG is equipped with a local PI-based controller with an anti-windup compensator and a distributed consensus-based global controller, from which a nonlinear networked error dynamics model is derived. The CPL nonlinearity is characterized via sector-boundedness with the S-procedure applied directly to yield tight LMI conditions, while the VSC input saturation is handled via a dead-zone decomposition and sector-boundedness, with both nonlinearities simultaneously absorbed into the dissipativity analysis. Both nonlinearities are simultaneously absorbed into the dissipativity analysis using the S-procedure. Subsequently, local controller gains and passivity indices, and distributed controller gains and the communication topology are co-designed by solving a sequence of local and global Linear Matrix Inequality (LMI) problems, enabling a one-shot co-design process that avoids iterative procedures. The effectiveness of the proposed framework is validated through simulation of an islanded DC MG under multiple operating scenarios, demonstrating robust performance superior to conventional control approaches.
\end{abstract}


\section{Introduction} 
The rapid proliferation of renewable energy resources and DC-native loads has fundamentally reshaped modern power systems, positioning DC microgrids (DC MGs) as a cornerstone technology for next-generation distributed energy infrastructure \cite{liu2023resilient}. Unlike AC systems, DC MGs eliminate unnecessary power conversion stages and the complexity of frequency regulation, resulting in superior efficiency, reduced hardware costs, and a simplified system architecture \cite{dou2022distributed}. Applications such as data centers, electric vehicle charging stations, LED lighting, and consumer electronics have further accelerated the adoption of DC MGs as a practical and economically viable solution for modern energy networks. Nevertheless, the fast closed-loop dynamics inherent to DC MGs place significant demands on the control system, and the presence of heterogeneous loads and distributed generation units introduces nonlinearities and uncertainties that render classical control techniques insufficient. Consequently, the development of robust, scalable, and theoretically rigorous control frameworks for DC MGs has attracted considerable research interest in the power systems and control communities \cite{guerrero2010hierarchical, khorsandi2014decentralized}.

The two primary control objectives in DC MGs are voltage regulation and current sharing. Centralized control strategies, while effective in achieving these objectives, suffer from a single point of failure and lack scalability \cite{guerrero2010hierarchical}. Decentralized approaches eliminate the need for inter-DG communication, yet the lack of coordination fundamentally compromises current-sharing accuracy \cite{khorsandi2014decentralized}. Distributed control, by contrast, allows each DG to exchange state information with its neighbors over a communication network and has demonstrated a superior ability to achieve both objectives simultaneously \cite{dehkordi2016distributed}. Within distributed control, droop-based methods remain the most widely adopted solution in the literature. However, droop control introduces an inherent trade-off between voltage regulation and current sharing, and its performance is sensitive to line impedance mismatch and droop coefficient selection \cite{nasirian2014distributed}. These fundamental limitations have motivated a shift toward droop-free distributed control strategies that rely entirely on inter-DG communication to coordinate voltage and current objectives without the structural compromises inherent in droop mechanisms \cite{dissanayake2019droop, zhang2022droop}.

ZIP loads, and specifically their CPL component, introduce a destabilizing nonlinearity due to the negative incremental impedance characteristic, posing significant challenges to stability analysis and controller synthesis \cite{kwasinski2010dynamic}, \cite{hassan2022dc}. Our prior work \cite{najafirad2025dissipativity} addressed CPL nonlinearities within a dissipativity-based co-design framework by incorporating sector-boundedness via the S-procedure indirectly, which resulted in conservative LMI conditions with loose upper bounds on the CPL sector. The current paper improves upon this by applying the S-procedure directly to the CPL sector constraint, yielding tighter bounds and a less conservative design. Moreover, the elevated control effort demanded by CPL compensation tends to push VSC inputs toward their saturation limits, making the treatment of input saturation a natural and necessary extension of the CPL handling problem.

Beyond CPLs, the VSCs that interface each DG to the DC bus introduce an additional, practically significant nonlinearity: input saturation. In practice, every VSC operates within a finite voltage command range imposed by its physical design and safety constraints, and the voltage command signal applied to each DG is therefore subject to hard amplitude limits \cite{vafamand2022dual}. When these limits are active, the actual control input deviates from its intended value, and a controller that does not explicitly account for this discrepancy can exhibit significant performance degradation or even closed-loop instability \cite{mardani2018design}. The treatment of input saturation as a sector-bounded nonlinearity via a dead-zone decomposition has been established as a systematic and LMI-compatible approach in the broader nonlinear control literature \cite{huff2021stability}. In practice, however, input saturation in PI-based controllers inevitably triggers integrator windup, necessitating an explicit anti-windup mechanism to prevent unbounded accumulation of the integral state and ensure closed-loop stability. Despite this, the rigorous co-design of anti-windup compensators alongside distributed controllers and communication topologies within a dissipativity-based framework for networked DC MGs remains an open problem that this paper addresses.

Significant research has addressed these two nonlinearities separately. For CPL stabilization, passivity-based control \cite{hassan}, sliding mode control \cite{alipour2022observer}, and various other nonlinear techniques \cite{hassan2022dc} have been proposed, though these works primarily consider single-converter or decentralized settings and do not address distributed current sharing or communication topology design. For VSC input saturation, dead-zone decomposition combined with sector-boundedness has been established as an LMI-compatible approach in the broader control literature \cite{huff2021stability}, yet its integration into a distributed co-design framework for networked DC MGs remains largely unexplored. Handling both nonlinearities simultaneously within a unified dissipativity-based co-design framework is the gap this paper addresses.

The joint co-design of distributed controller parameters and communication topology is another challenge in DC MG control, one that our prior work \cite{najafirad2025dissipativity} addressed within a dissipativity-based framework. While co-design has intrinsic benefits such as avoiding over-provisioned communication infrastructure and exploiting the coupling between network connectivity and closed-loop performance \cite{hu2021cost, lou2018optimal}, it also offers a distinct advantage in the context of nonlinear DC MGs. Specifically, the additional degrees of freedom introduced by treating the topology as a design variable provide multiple avenues to absorb and compensate for the destabilizing effects of nonlinearities such as CPLs and VSC input saturation. This flexibility motivates retaining the co-design structure in the present paper, where it is extended to simultaneously and rigorously handle both nonlinearities within a unified dissipativity-based framework.

To address the aforementioned challenges in a unified and rigorous manner, this paper adopts dissipativity theory as the central analytical framework. Dissipativity theory offers a powerful and flexible foundation for the analysis and synthesis of robust control systems for large-scale networked systems, as it characterizes the input-output energy behavior of each subsystem without requiring explicit knowledge of the global system dynamics \cite{arcak2022}. By focusing on the fundamental energy exchanges between interconnected subsystems, dissipativity-based approaches can certify stability and robustness of the overall networked system even when individual subsystems exhibit complex nonlinear behaviors \cite{datadriven}. Crucially, the dissipativity framework is compatible with LMI-based synthesis tools, which enable the formulation of computationally tractable convex optimization problems for controller design. The sector-bounded characterizations of both the CPL nonlinearity and the VSC input saturation are naturally accommodated within this framework via the S-procedure and Young's inequality, enabling simultaneous treatment of both nonlinearities without sacrificing convexity or introducing conservatism in the resulting synthesis conditions. This combination of theoretical rigor, computational tractability, and ability to handle multiple nonlinearities makes dissipativity theory a particularly well-suited tool for the problem considered in this paper.


In this paper, the DC MG is modeled as a networked system comprising DGs, transmission lines, and ZIP loads interconnected via a static matrix, where each DG is equipped with a hierarchical controller including a local PI-based controller with an anti-windup compensator and a distributed consensus-based global controller. A hierarchical controller comprising steady-state, local, and distributed global components is developed for each DG, and a thorough equilibrium analysis is conducted to identify the operating point and establish the feasibility conditions on the saturation constraints. The closed-loop error dynamics of the DC MG are then derived as a nonlinear networked error system that incorporates both disturbance inputs and regulated performance outputs to ensure robust achievement of the desired control objectives. The CPL nonlinearity and the VSC input saturation are each characterized via sector-boundedness, and both are simultaneously incorporated into the dissipativity analysis through the S-procedure and Young's inequality, certifying an input feedforward output feedback passivity (IF-OFP) property for each DG subsystem. Building on this, the local controller gains and  passivity indices, and distributed controller gains and the communication topology are co-designed by solving a sequence of local LMI problems followed by a global LMI problem. The local and global LMI problems are solved in a tractable one-shot co-design procedure that avoids iterative schemes \cite{almihat2023overview}.


The present paper is part of a broader research program on dissipativity-based distributed control of microgrids. Prior works in this line established the co-design framework for DC MGs with progressively richer features, from basic voltage regulation \cite{ACCNajafi} to ZIP load handling and current sharing \cite{najafirad2025dissipativity}, and most recently to AC MGs \cite{najafirad2025dissipativityACMG}. The present paper advances this line by introducing VSC input saturation with anti-windup compensation and applying the S-procedure directly to the CPL sector constraint for tighter LMI conditions.

The main contributions of this paper can be summarized as follows:
\begin{enumerate}
\item We extend the dissipativity-based co-design framework established in our prior work, which addressed DC MG networked system modeling, multi-objective control, hierarchical control architecture, and one-shot control-topology co-design, to remain valid and rigorous in the presence of CPL and VSC input saturation nonlinearities.
\item We improve upon our prior treatment of CPL nonlinearities by applying the S-procedure directly to the sector constraint, yielding tighter LMI conditions and a less conservative design compared to the indirect approach used in \cite{najafirad2025dissipativity}.
\item We explicitly model VSC input saturation via a dead-zone decomposition and establish its sector-bounded characterization, incorporating it into the dissipativity analysis simultaneously with the CPL nonlinearity within a unified framework.
\item We rigorously integrate an anti-windup compensator into the local controller design and analysis process, explicitly accounting for its effect within the LMI-based synthesis to prevent integrator windup under saturating conditions.
\end{enumerate}

The remainder of this paper is structured as follows. The essential concepts of dissipativity theory and networked systems are presented in Sec. \ref{Preliminaries} to establish the theoretical foundation. The DC MG model with detailed physical topology and component dynamics is introduced in Sec. \ref{problemformulation}. A novel hierarchical control architecture that eliminates traditional droop mechanisms is developed in Sec. \ref{Sec:Controller}, where the equilibrium analysis and VSC saturation feasibility conditions are also established. The nonlinear networked error dynamics, incorporating both CPL and VSC input saturation nonlinearities, are derived in Sec. \ref{Sec:ControlDesign}. The dissipativity-based methodology for controller and communication topology co-design, including the extended LMI formulations that account for both nonlinearities, is presented in Sec. \ref{Passivity-based Control}. Numerical simulations demonstrating the effectiveness of the proposed framework under multiple operating scenarios, including saturation-active conditions, are provided in Sec. \ref{Simulation}. Finally, Sec. \ref{Conclusion} offers concluding remarks and directions for future research.

\section{Preliminaries}\label{Preliminaries}

\subsection{Notations}
The notation $\mathbb{R}$ and $\mathbb{N}$ signify the sets of real and natural numbers, respectively. 
For any $N\in\mathbb{N}$, we define $\mathbb{N}_N\triangleq\{1,2,..,N\}$.
An $n \times m$ block matrix $A$ is denoted as $A = [A_{ij}]_{i \in \mathbb{N}_n, j \in \mathbb{N}_m}$. Either subscripts or superscripts are used for indexing purposes, e.g., $A_{ij} \equiv A^{ij}$.
$[A_{ij}]_{j\in\mathbb{N}_m}$ and $\diag([A_{ii}]_{i\in\mathbb{N}_n})$ represent a block row matrix and a block diagonal matrix, respectively.
$\0$ and $\I$, respectively, are the zero and identity matrices (dimensions will be clear from the context). A symmetric positive definite (semi-definite) matrix $A\in\mathbb{R}^{n\times n}$ is denoted by $A>0\ (A\geq0)$. The symbol $\star$ represents conjugate blocks inside block symmetric matrices. $\mathcal{H}(A)\triangleq A + A^\T$,  $\mb{1}_{\{ \cdot \}}$ is the indicator function and $\mathbf{1}_N$ is a vector in $\R^N$ containing only ones.

\subsection{Dissipativity}
Consider a nonlinear dynamic system:
\begin{equation}\label{dynamic}
\begin{aligned}
    \dot{x}(t)=f(x(t),u(t)),\\
    y(t)=h(x(t),u(t)),
    \end{aligned}
\end{equation}
where $x(t)\in\mathbb{R}^n$, $u(t)\in\mathbb{R}^q$, $y(t)\in\mathbb{R}^m$, and $f:\mathbb{R}^n\times\mathbb{R}^q\rightarrow\mathbb{R}^n$ and $h:\mathbb{R}^n\times\mathbb{R}^q\rightarrow\mathbb{R}^m$ are continuously differentiable and $f(\0,\0)=\0$ and $h(\0,\0)=\0$.


\begin{definition}\cite{arcak2022}
The system \eqref{dynamic} is dissipative under supply 
rate $s:\mathbb{R}^q\times\mathbb{R}^m\rightarrow\mathbb{R}$ 
if there exists a continuously differentiable storage 
function $V:\mathbb{R}^n\rightarrow\mathbb{R}$ such that 
$V(x)>0$, $\forall x \neq \mathbf{0}$, $V(\mathbf{0})=0$, 
and $\dot{V}(x)=\nabla_x V(x)f(x,u) \leq s(u,y)$, 
for all $(x,u)\in\mathbb{R}^n\times\mathbb{R}^q$.
\end{definition}

\begin{definition}
The system (\ref{dynamic}) is $X$-dissipative if it is dissipative under the quadratic supply rate:
\begin{center}
$
s(u,y)\triangleq
\begin{bmatrix}
    u \\ y
\end{bmatrix}^\top
\begin{bmatrix}
    X^{11} & X^{12}\\ X^{21} & X^{22}
\end{bmatrix}
\begin{bmatrix}
    u \\ y
\end{bmatrix}.
$
\end{center}
\end{definition}

\begin{remark}\label{Rm:X-DissipativityVersions}
The system \eqref{dynamic} is $X$-dissipative with 
specific properties depending on the structure of $X$ \cite{welikala2022line}:\\
1) passive, if $X = \begin{bmatrix} \mathbf{0} & \frac{1}{2}\mathbf{I} \\ 
\frac{1}{2}\mathbf{I} & \mathbf{0} \end{bmatrix}$;\\
2) input feedforward and output feedback passive with indices $\nu$ and $\rho$, i.e., IF-OFP($\nu,\rho$), if $X = \begin{bmatrix} -\nu\mathbf{I} & 
\frac{1}{2}\mathbf{I} \\ \frac{1}{2}\mathbf{I} & -\rho\mathbf{I} 
\end{bmatrix}$;\\
3) $L_2$-stable with gain $\gamma$, i.e., L2G($\gamma$), if $X = \begin{bmatrix} \gamma^2\mathbf{I} & 
\mathbf{0} \\ \mathbf{0} & -\mathbf{I} \end{bmatrix}$;
\end{remark}

If the system \eqref{dynamic} is linear time-invariant (LTI), a necessary and sufficient condition for $X$-dissipativity is provided in the following proposition as a linear matrix inequality (LMI) problem.

\begin{proposition}\label{Prop:linear_X-EID} \cite{welikala2023platoon}
The LTI system
\begin{equation*}\label{Eq:Prop:linear_X-EID_1}
\begin{aligned}
    \dot{x}(t)=Ax(t)+Bu(t),\quad
    y(t)=Cx(t)+Du(t),
\end{aligned}
\end{equation*}
is $X$-dissipative if and only if there exists $P>0$ such that
\begin{equation*}\label{Eq:Prop:linear_X-EID_2}
\scriptsize
\begin{bmatrix}
-\mathcal{H}(PA)+C^\top X^{22}C & -PB+C^\top X^{21}+C^\top X^{22}D\\
\star & X^{11}+\mathcal{H}(X^{12}D)+D^\top X^{22}D
\end{bmatrix}
\normalsize
\geq0.
\end{equation*}
\end{proposition}



\subsection{Networked Systems} \label{SubSec:NetworkedSystemsPreliminaries}

Consider the networked system $\Sigma$ in Fig. \ref{Networked}, consisting of dynamic subsystems $\Sigma_i,i\in\mathbb{N}_N$, $\Bar{\Sigma}_i,i\in\mathbb{N}_{\Bar{N}}$ and a static interconnection matrix $M$ that characterizes interconnections among subsystems, exogenous inputs $w(t)\in\mathbb{R}^r$ (e.g. disturbances) and interested outputs $z(t)\in\mathbb{R}^l$ (e.g. performance). 

\begin{figure}
    \centering
    \includegraphics[width=0.6\columnwidth]{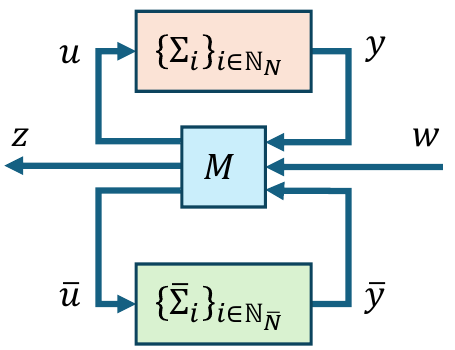}
    \caption{A generic networked system $\Sigma$.}
    \label{Networked}
\end{figure}

The dynamics of each subsystem $\Sigma_i,i\in\mathbb{N}_N$ are given by
\begin{equation}
    \begin{aligned}
        \dot{x}_i(t)=f_i(x_i(t),u_i(t)), \quad y_i(t)=h_i(x_i(t),u_i(t)),
    \end{aligned}
\end{equation}
where $x_i(t)\in\mathbb{R}^{n_i}$, $u_i(t)\in\mathbb{R}^{q_i}$, $y_i(t)\in\mathbb{R}^{m_i}$ and $f_i(\0,\0)=\0$ and $h_i(\0,\0)=\0$. In addition, each subsystem $\Sigma_i$ is assumed to be $X_i$-dissipative , where $X_i \triangleq [X_i^{kl}]_{k,l\in\N_2}$.
Regarding each subsystem $\bar{\Sigma}_i, i\in\N_{\bar{N}}$, we use similar assumptions and notations, but include a bar symbol to distinguish between the two types of subsystems, e.g.,  $\bar{\Sigma}_i$ is assumed to be $\bar{X}_i$-dissipative where $\bar{X}_i \triangleq [\bar{X}_i^{kl}]_{k,l\in\N_2}$.

Defining $u\triangleq[u_i^\top]^\top_{i\in\mathbb{N}_N}$, $y\triangleq[y_i^\top]^\top_{i\in\mathbb{N}_N}$, $\Bar{u}\triangleq[\Bar{u}_i^\top]^\top_{i\in\mathbb{N}_{\Bar{N}}}$ and $\Bar{y}\triangleq[y_i^\top]^\top_{i\in\mathbb{N}_{\Bar{N}}}$, the interconnection matrix $M$ and the corresponding interconnection relationship are given by
\begin{equation}\label{interconnectionMatrix}
\scriptsize
\begin{bmatrix}
    u \\ \bar{u} \\ z
\end{bmatrix}=M
\normalsize
\scriptsize
\begin{bmatrix}
    y \\ \bar{y} \\ w
\end{bmatrix}
\normalsize
\equiv
\scriptsize
\begin{bmatrix}
    M_{uy} & M_{u\bar{y}} & M_{uw}\\
    M_{\bar{u}y} & M_{\bar{u}\bar{y}} & M_{\bar{u}w}\\
    M_{zy} & M_{z\bar{y}} & M_{zw}
\end{bmatrix}
\begin{bmatrix}
    y \\ \bar{y} \\ w
\end{bmatrix}.
\normalsize
\end{equation}

The following proposition exploits the $X_i$-dissipative and $\bar{X}_i$-dissipative properties of the subsystems $\Sigma_i,i\in\mathbb{N}_N$ and $\Bar{\Sigma}_i,i\in\mathbb{N}_{\Bar{N}}$ to formulate an LMI problem for synthesizing the interconnection matrix $M$ (\ref{interconnectionMatrix}), ensuring the networked system $\Sigma$ is $\textbf{Y}$-dissipative for a prespecified/optimized $\textbf{Y}$ under two mild assumptions (see \cite[Rm. 4,5]{welikala2023non}).

\begin{assumption}\label{As:NegativeDissipativity}
    For the networked system $\Sigma$, the interested \textbf{Y}-dissipative specification is such that $\textbf{Y}^{22}<0$.
\end{assumption}


\begin{assumption}\label{As:PositiveDissipativity}
    In the networked system $\Sigma$, each subsystem $\Sigma_i$ is $X_i$-dissipative with $X_i^{11}>0, \forall i\in\mathbb{N}_N$, and similarly, each subsystem $\bar{\Sigma}_i$ is $\bar{X}_i$-dissipative with $\bar{X}_i^{11}>0, \forall i\in\N_{\bar{N}}$.
\end{assumption}



\begin{proposition}\label{synthesizeM}\cite{welikala2023non}
    Under As. \ref{As:NegativeDissipativity}-\ref{As:PositiveDissipativity}, the network system $\Sigma$ can be made \textbf{Y}-dissipative (from $w(t)$ to $z(t)$) by synthesizing the interconnection matrix $M$ \eqref{interconnectionMatrix} via solving the LMI problem:
\begin{equation}
\begin{aligned}
	&\mbox{Find: } 
	L_{uy}, L_{u\bar{y}}, L_{uw}, L_{\bar{u}y}, L_{\bar{u}\bar{y}}, L_{\bar{u}w}, M_{zy}, M_{z\bar{y}}, M_{zw}, \\
	&\mbox{Sub. to: } p_i \geq 0, \forall i\in\N_N, \ \ 
	\bar{p}_i \geq 0, \forall i\in\N_{\bar{N}},\ \text{and} \  \eqref{NSC4YEID},
\end{aligned}
\end{equation}
with
\scriptsize
$\bm{M_{uy} & M_{u\bar{y}} & M_{uw} \\ M_{\bar{u}y} & M_{\bar{u}\bar{y}} & M_{\bar{u}w}} = 
\bm{\textbf{X}_p^{11} & \0 \\ \0 & \bar{\textbf{X}}_{\bar{p}}^{11}}^{-1} \hspace{-1mm} \bm{L_{uy} & L_{u\bar{y}} & L_{uw} \\ L_{\bar{u}y} & L_{\bar{u}\bar{y}} & L_{\bar{u}w}}$.
\normalsize
where $\textbf{X}_p^{kl} \triangleq \diag(\{p_iX_i^{kl}:i\in\N_N\}), \forall k,l\in\N_2$, $\textbf{X}^{12} \triangleq \diag((X_i^{11})^{-1}X_i^{12}:i\in\N_N)$, and 
$\textbf{X}^{21} \triangleq (\textbf{X}^{12})^\T$ (terms $\bar{\textbf{X}}_{\bar{p}}^{kl}$, $\bar{\textbf{X}}^{12}$ and 
$\bar{\textbf{X}}^{21}$ have analogous definitions).
\end{proposition}


\begin{figure*}[!hb]
\vspace{-5mm}
\centering
\hrulefill
\begin{equation}\label{NSC4YEID}
\scriptsize
 \bm{
		\textbf{X}_p^{11} & \0 & \0 & L_{uy} & L_{u\bar{y}} & L_{uw} \\
		\0 & \bar{\textbf{X}}_{\bar{p}}^{11} & \0 & L_{\bar{u}y} & L_{\bar{u}\bar{y}} & L_{\bar{u}w}\\
		\0 & \0 & -\textbf{Y}^{22} & -\textbf{Y}^{22} M_{zy} & -\textbf{Y}^{22} M_{z\bar{y}} & \textbf{Y}^{22} M_{zw}\\
		L_{uy}^\T & L_{\bar{u}y}^\T & - M_{zy}^\T\textbf{Y}^{22} & -L_{uy}^\T\textbf{X}^{12}-\textbf{X}^{21}L_{uy}-\textbf{X}_p^{22} & -\textbf{X}^{21}L_{u\bar{y}}-L_{\bar{u}y}^\T \bar{\textbf{X}}^{12} & -\textbf{X}^{21}L_{uw} + M_{zy}^\T \textbf{Y}^{21} \\
		L_{u\bar{y}}^\T & L_{\bar{u}\bar{y}}^\T & - M_{z\bar{y}}^\T\textbf{Y}^{22} & -L_{u\bar{y}}^\T\textbf{X}^{12}-\bar{\textbf{X}}^{21}L_{\bar{u}y} & 		-(L_{\bar{u}\bar{y}}^\T \bar{\textbf{X}}^{12} + \bar{\textbf{X}}^{21}L_{\bar{u}\bar{y}}+\bar{\textbf{X}}_{\bar{p}}^{22}) & -\bar{\textbf{X}}^{21} L_{\bar{u}w} + M_{z\bar{y}}^\T \textbf{Y}^{21} \\ 
		L_{uw}^\T & L_{\bar{u}w}^\T & -M_{zw}^\T \textbf{Y}^{22}& -L_{uw}^\T\textbf{X}^{12}+\textbf{Y}^{12}M_{zy} & -L_{\bar{u}w}^\T\bar{\textbf{X}}^{12}+ \textbf{Y}^{12} M_{z\bar{y}} & M_{zw}^\T\textbf{Y}^{21} + \textbf{Y}^{12}M_{zw} + \textbf{Y}^{11}
	}\normalsize > 0
\end{equation}
\end{figure*} 

\subsection{Linear Algebraic Preliminaries}

Before concluding this section, we recall three linear algebraic results that will be useful in the sequel.

\begin{lemma}\label{Lm:Schur_comp}
\textbf{(Schur Complement)} For matrices $P > 0, Q$ and $R$, the following statements are equivalent:
\begin{subequations}
\begin{align}
1)\ &\begin{bmatrix} P & Q \\ Q^\T & R \end{bmatrix} > 0, \label{Lm:Schur1}\\
2)\ &P > 0, R-Q^\T P^{-1}Q > 0, \label{Lm:Schur2}\\
3)\ &R > 0, P-QR^{-1}Q^\T > 0 \label{Lm:Schur3}.
\end{align}
\end{subequations}
\end{lemma}

\begin{proof}
Let $z = \begin{bmatrix} z_1 & z_2 \end{bmatrix}^\top$ be any non-zero vector and $y = z_2 + P^{-1}Qz_1$. Then:
\begin{align}
    &z^\top\begin{bmatrix} P & Q \\ Q^\top & R \end{bmatrix}z \nonumber\\
    &= z_1^\top P z_1 + z_1^\top Q z_2 
       + z_2^\top Q^\top z_1 + z_2^\top R z_2 \nonumber\\
    &= z_1^\top(P - QP^{-1}Q^\top)z_1 
       + y^\top(R - Q^\top P^{-1}Q)y \nonumber\\
    &\geq y^\top(R - Q^\top P^{-1}Q)y, \nonumber
\end{align}
where the third step follows by completing the square using $y$, and the last step holds since $P > 0$ implies $P - QP^{-1}Q^\top \geq 0$. Therefore, the expression is non-negative if and only if $R - Q^\top P^{-1}Q \geq 0$, establishing the equivalence of \eqref{Lm:Schur1} and \eqref{Lm:Schur2}. The equivalence of \eqref{Lm:Schur1} and \eqref{Lm:Schur3} follows by symmetric arguments.
\end{proof}

\begin{lemma}\label{Lm:InvSchur}
For any $P > 0$ and a square matrix $Q$:  
$$Q^\T P^{-1}Q \geq Q^\T + Q - P.$$ 
\end{lemma}
\begin{proof}
For any arbitrary matrix $S$, since $P>0$: we have 
$(S-\I)^\T P (S-\I) \geq 0$, which simplifies to  
$$S^\T P S - PS - S^\T P + P \geq 0.$$
The required result follows by applying the change of 
variables $S = P^{-1}Q$ and rearranging the terms, 
with equality holding if and only if $Q = P$.
\end{proof}


\begin{lemma}\label{Lm:Woodbury}
For an invertible $R \in \mathbb{R}^{n \times n}$ and $\rho \in \R_{>0}$:
$$
(R + \rho \I)^{-1} = R^{-1} - \rho R^{-1} \left( \I + \rho R^{-1} \right)^{-1} R^{-1}.
$$
\end{lemma}
\begin{proof}
The result follows directly from applying the well-known Woodbury Matrix Identity \cite{horn2012matrix}:
$$
(R + UV^T)^{-1} = R^{-1} - R^{-1}U \left( \I + V^T R^{-1}U \right)^{-1} V^T R^{-1}
$$
with the choices $U = \sqrt{\rho} \I$ and $V = \sqrt{\rho} \I$.
\end{proof}

\begin{lemma}(\textbf{Matrix S-Lemma}, 
\cite{datadriven})\label{Lm:Sprocedure}
Let $\Pi$ and $\Gamma$ be symmetric matrices. 
If there exists $\zeta$ such that 
$\zeta^\top\Gamma\zeta > 0$, then:
\begin{equation}\label{Eq:Sprocedure}
\zeta^\top\Gamma\zeta \geq 0 
\implies 
\zeta^\top\Pi\zeta \geq 0,
\quad \forall \zeta,
\end{equation}
if and only if there exists $\lambda \geq 0$ 
such that:
\begin{equation}\label{Eq:Sprocedure_LMI}
\Pi - \lambda\Gamma \geq 0.
\end{equation}
\end{lemma}

\section{Problem Formulation}\label{problemformulation}
This section presents the dynamic modeling of the DC MG, which consists of multiple DGs, loads, and transmission lines. Specifically, our modeling approach is motivated by \cite{nahata}, which highlights the role and impact of communication and physical topologies in DC MGs.

\subsection{DC MG Physical Interconnection Topology}
The physical interconnection topology of a DC MG is modeled as a directed connected graph $\mathcal{G}^p =(\mathcal{V},\mathcal{E})$ where $\mathcal{V} \triangleq \mathcal{D} \cup \mathcal{L}$ is bipartite: $\mathcal{D} \triangleq \{\Sigma_i^{DG}, i\in\N_N\}$ (DGs) and $\mathcal{L} \triangleq \{\Sigma_l^{line}, l\in\N_L\}$ (lines). The DGs are interconnected with each other through transmission lines. The interface between each DG and the DC MG is through a point of common coupling (PCC). For simplicity, the loads are assumed to be connected to the DG terminals at the respective PCCs \cite{dorfler2012kron}. Indeed loads can be moved to PCCs using Kron reduction even if they are located elsewhere \cite{dorfler2012kron}.


To represent the DC MG's physical topology, we use its bi-adjacency matrix $\mathcal{A} \triangleq  \scriptsize 
\begin{bmatrix}
\0 & \mathcal{B} \\
\mathcal{B}^\T & \0
\end{bmatrix},
\normalsize$  
where $\mathcal{B} \in \mathbb{R}^{N \times L}$ is the bi-adjacency matrix of $\mathcal{G}^p$. In particular, $\mathcal{B} \triangleq [\mathcal{B}_{il}]_{i \in \N_N, l \in \N_L}$ with
$\mathcal{B}_{il} \triangleq \mb{1}_{\{l\in\mathcal{E}_i^+\}} - \mb{1}_{\{l\in\mathcal{E}_i^-\}},$ where $\mathcal{E}_i^+$ and $\mathcal{E}_i^-$ represent the sets of out- and in-edges, respectively, under an arbitrary but fixed orientation assigned to the edges of $\mathcal{G}^p$.
Here, $\mathcal{E}_i \triangleq \mathcal{E}_i^+ \cup \mathcal{E}_i^-$ denotes the set of all transmission lines connected to $\Sigma_i^{DG}$, $\forall i \in \mathbb{N}_N$, and $\mathcal{E}_l \triangleq \{i \in \mathbb{N}_N : l \in \mathcal{E}_i\}$ denotes the set of DGs connected to transmission line 
$\Sigma_l^{line}$, $\forall l \in \mathbb{N}_L$.


\subsection{Dynamic Model of a Distributed Generator (DG)}
Each DG consists of a DC voltage source, a voltage source converter (VSC), and some RLC components. Each DG $\Sigma_i^{DG},i\in\N_N$ supplies power to a specific ZIP load at its PCC (denoted $\text{PCC}_i$). Additionally, it interconnects with other DG units via transmission lines $\{\Sigma_l^{line}:l \in \mathcal{E}_i\}$. Figure \ref{DCMG} illustrates the schematic diagram of $\Sigma_i^{DG}$, including the local load, a connected transmission line, and the steady state, local, and distributed global controllers.

By applying Kirchhoff's Current Law (KCL) and Kirchhoff's Voltage Law (KVL) at $\text{PCC}_i$ on the DG side, we get the following equations for $\Sigma_i^{DG},i\in\N_N$:
\begin{equation}
\begin{aligned}\label{DGEQ}
\Sigma_i^{DG}:
\begin{cases}
    C_{ti}\frac{dV_i}{dt} &= I_{ti} - I_{Li}(V_i) - I_i + w_{vi}, \\
L_{ti}\frac{dI_{ti}}{dt} &= -V_i - R_{ti}I_{ti} + \text{sat}(V_{ti}) + w_{ci},
\end{cases}
\end{aligned}
\end{equation}
where the parameters $R_{ti}$, $L_{ti}$, and $C_{ti}$ 
represent the internal resistance, internal inductance, and filter capacitance of $\Sigma_i^{DG}$, respectively. The state variables are selected as $V_i$ and $I_{ti}$, where $V_i$ is the $\text{PCC}_i$ voltage and $I_{ti}$ is the internal current. $V_{ti}$ is the input command signal applied to the VSC and $\text{sat}(V_{ti})$ is the VSC saturation function. $I_{Li}(V_i)$ is the total current drawn by the local load and $I_i$ is the total current injected to the DC MG by $\Sigma_i^{DG}$, given by:
\begin{equation}
\label{Eq:DGCurrentNetOut}
I_i
=\sum_{l\in\mathcal{E}_i}\mathcal{B}_{il}I_l,
\end{equation}
where $I_l$, $l\in\mathcal{E}_i$ are line currents. We have also included $w_{vi}$ and $w_{ci}$ terms in \eqref{DGEQ} to represent 
unknown disturbances (assumed bounded and zero mean) resulting from external effects or modeling imperfections.


\subsection{Dynamic Model of a Transmission Line}
As shown in Fig. \ref{DCMG}, the power line $\Sigma_l^{line}$ can be represented as an RL circuit with resistance $R_l$ and inductance $L_l$. By applying KVL to $\Sigma_l^{line}$, we obtain:
\begin{equation}\label{line}
    \Sigma_l^{line}: 
    \begin{cases}  
        L_l\frac{dI_l}{dt}=-R_lI_l+\Bar{u}_l + \bar{w}_l,
    \end{cases}
\end{equation}
where $I_l$ is the line current (i.e., the state), $\Bar{u}_l=V_i-V_j=\sum_{i\in \mathcal{E}_l}\mathcal{B}_{il}V_i$ is the voltage differential (i.e., the line input), and $\bar{w}_l(t)$ represents the unknown disturbance (assumed bounded and zero mean) that affects the line dynamics.


\begin{figure}
    \centering
    \includegraphics[width=1\columnwidth]{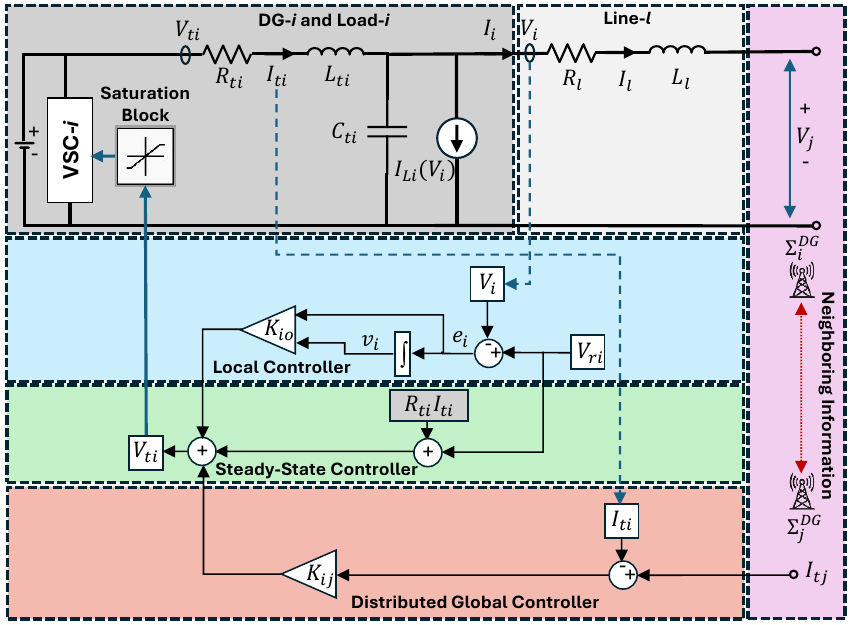}
    \caption{The electrical schematic of DG-$i$, load-$i$, $i\in\N_N$, local controller, distributed global controller, and line-$l$, $l\in\N_L$.}
    \label{DCMG}
\end{figure}

\subsection{Dynamic Model of a ZIP Load} 
Recall that $I_{Li}(V_i)$ in \eqref{DGEQ} (see also Fig. \ref{DCMG}) is the total current drawn by the load at $\Sigma_i^{DG}, i\in\N_N$. As the load is assumed to be a generic ``ZIP'' load, $I_{Li}(V_i)$ takes the form:
\begin{equation}\label{Eq:LoadModel}
I_{Li}(V_i) = I_{Li}^Z(V_i) + I_{Li}^I(V_i) + I_{Li}^P(V_i).
\end{equation}

Here, the ZIP load's components are: 
(i) a constant impedance load: $I_{Li}^{Z}(V_i)=Y_{Li}V_i$, where  $Y_{Li}=1/R_{Li}$ is the conductance of this load component;
(ii) a constant current load: $I_{Li}^{I}(V_i)=\bar{I}_{Li}$, where $\bar{I}_{Li}$ is the current demand of this load component; and
(iii) a constant power load (CPL): $I_{Li}^{P}(V_i)=V_i^{-1}P_{Li}$, where $P_{Li}$ represents the power demand of this load component. 

\subsection{Nonlinearities in DC MGs}
As opposed to $I_{Li}^{Z}(V_i)$ and $I_{Li}^{I}(V_i)$ (that take an affine linear form in the chosen state variables), the CPL $I_{Li}^{P}(V_i)$ introduces significant stability challenges due to its inherent negative impedance characteristic. This can be observed by examining the small-signal impedance of the CPL:
\begin{equation*}
    Z_{CPL} = \frac{\partial V_i}{\partial I_{Li}^P} = \frac{\partial V_i}{\partial (P_{Li}/V_i)} = -\frac{V_i^2}{P_{Li}} < 0.
\end{equation*}

This negative impedance characteristic creates a destabilizing effect in the DC MG, as it tends to amplify voltage perturbations rather than dampen them \cite{emadi2006constant}. When a small voltage drop occurs, the CPL draws more current to maintain constant power, further reducing the voltage and potentially leading to voltage collapse if not properly controlled.

The nonlinear nature of CPLs also introduces complexities for the control design. In particular, the nonlinear term $I_{Li}^P(V_i) = V_i^{-1}P_{Li}$ appears in the voltage dynamics (not in current dynamics) channel in \eqref{DGEQ}, and hence cannot be directly canceled using state feedback linearization techniques.  
Consequently, this nonlinearity must be carefully accounted for to ensure system stability and robustness, as often CPLs constitute a significant portion of the total ZIP load. 

The second nonlinearity arises from VSC input saturation, 
where the voltage command $V_{ti}$ applied to each DG is 
subject to hard amplitude limits imposed by the physical 
design and safety constraints of the VSC:
\begin{equation*}
\text{sat}(V_{ti}) \triangleq
\begin{cases}
    V_{ti}^{\max} & \text{if} \ V_{ti} > V_{ti}^{\max},\\
    V_{ti} &  \text{if} \ V_{ti}^{\min} \leq V_{ti} \leq V_{ti}^{\max},\\
    V_{ti}^{\min} & \text{if} \ V_{ti} < V_{ti}^{\min},
\end{cases}
\end{equation*}
with $V_{ti}^{\min}$ and $V_{ti}^{\max}$ denoting the lower 
and upper saturation limits of the VSC of $\Sigma_i^{DG}$, 
$i \in \mathbb{N}_N$, respectively. When these limits are 
active, the actual control input deviates from its intended 
value, and a controller that does not explicitly account for 
this discrepancy can exhibit significant performance 
degradation or even closed-loop instability 
\cite{huff2021stability}. Moreover, the elevated control 
effort demanded by CPL compensation tends to push the VSC 
inputs closer to their saturation limits, making these two 
nonlinearities inherently coupled in practice. Additionally, 
saturation in the PI-based local controllers inevitably 
triggers integrator windup, causing unbounded accumulation 
of the integral state and further degrading closed-loop 
performance. As we will see in the subsequent sections, 
the proposed control framework handles VSC input saturation 
via a dead-zone decomposition and sector-bounded 
characterization, with an anti-windup compensator 
rigorously integrated into the controller design and 
analysis to address the windup phenomenon.

\section{Proposed Hierarchical Control Architecture}\label{Sec:Controller}

The primary control objective of the DC MG is to ensure that the $\text{PCC}_i$ voltage $V_i$ at each $\Sigma_i^{DG}, i\in\N_N $ closely follows a specified reference voltage $V_{ri}$ while maintaining a proportional current sharing among DGs (with respect to their power ratings). In the proposed control architecture, these control objectives are achieved through the complementary action of local and distributed controllers. The local controller at each $\Sigma_i^{DG}$ is a PI controller responsible for voltage regulation. On the other hand, the distributed global controller at each $\Sigma_i^{DG}$ is a consensus-based controller that ensures proper current sharing among DGs.

\subsection{Local Voltage Regulating Controller}
At each $\Sigma_i^{DG}, i\in\N_N$, for its PCC$_i$ voltage $V_i(t)$ to effectively track the assigned reference voltage $V_{ri}(t)$, it is imperative to ensure that the tracking error $e_i(t)\triangleq V_i(t)-V_{ri}(t)$ converges to zero, i.e. $\lim_{t \to \infty} (V_i(t) - V_{ri}) = 0$. To this end, motivated by \cite{tucci2017}, we first include each $\Sigma_i^{DG}, i\in\N_N$ with an integrator state $v_i$ defined as $v_i(t) \triangleq \int_0^t (V_i(\tau) - V_{ri})d\tau$ (see also Fig. \ref{DCMG}) that follows the dynamics 
\begin{equation}\label{error}
    \frac{dv_i(t)}{dt}=e_i(t)=V_i(t)-V_{ri} - K_{aw,i}\phi_i(u_i),
\end{equation}
where $u_i(t) \triangleq V_{ti}(t)$ is the overall control 
input applied to the VSC of $\Sigma_i^{DG}$, $K_{aw,i} > 0$ is the anti-windup gain, and $\phi_i(u_i)$ is the dead-zone 
nonlinearity defined as
\begin{equation}\label{Eq:Saturation}
\phi_i(u_i) \triangleq
    \begin{cases}
        V^{\max}_{ti} - u_i & \text{if}\ u_i > V^{\max}_{ti},\\
        0 & \text{if}\ V^{\min}_{ti} \leq u_i \leq V^{\max}_{ti},\\
        V^{\min}_{ti} - u_i & \text{if}\ u_i < V^{\min}_{ti}.
    \end{cases}
\end{equation}
When saturation is inactive, $\phi_i(u_i) = 0$ and \eqref{Eq:Saturation} 
reduces to the standard integrator, otherwise, $-K_{aw,i}\phi_i(u_i)$ prevents unbounded accumulation of $v_i$ \cite{tarbouriech2009anti}. Each $\Sigma_i^{DG}, i\in\N_N$ is equipped with a local state feedback PI controller:
\begin{equation}\label{Controller}
  u_{iL}(t)\triangleq  k_{i0}^P (V_i-V_{ri}) + k_{i0}^I v_i(t) = K_{i0}x_i(t) - k_{i0}^P V_{ri},
\end{equation}
where 
\begin{equation}\label{Eq:DGstate}
x_i \triangleq \begin{bmatrix}
    V_i &  I_{ti} & v_i
\end{bmatrix}^\top,
\end{equation}
denotes the augmented state (henceforth referred to as the state) of $\Sigma_i^{DG}$ and $K_{i0}\triangleq \begin{bmatrix}
    k_{i0}^P & 0 & k_{i0}^I
\end{bmatrix}\in\mathbb{R}^{1\times3}$ is the local controller gain matrix.

\subsection{Distributed Global Controller}
We implement distributed global controllers at each DG, and task them with maintaining a proportional current sharing among the DGs. In particular, their objective is to ensure:
\begin{equation}\label{Eq:PropCurrSharing}
\frac{I_{ti}(t)}{P_{ni}} = \frac{I_{tj}(t)}{P_{nj}} = I_s, \quad \forall i,j\in\mathbb{N}_N,
\end{equation}
where $P_{ni}$ and $P_{nj}$ represent the power ratings of DGs  $\Sigma_i^{DG}$ and $\Sigma_j^{DG}$ respectively, and $I_s$ represents the common current sharing ratio that emerges from balancing the total load demand among DGs according to their power ratings.

To address the current sharing, as shown in Fig. \ref{DCMG}, we employ a consensus-based distributed controller. The communication topology is denoted as a directed graph $\mathcal{G}^c = (\mathcal{D}, \mathcal{F})$ where 
$\mathcal{D} \triangleq \{\Sigma_i^{DG}, i\in\mathbb{N}_N\}$ 
and $\mathcal{F}$ represents the set of communication links 
among DGs, with $\mathcal{F}_i^+$ and $\bar{\mathcal{F}}_i^-$ 
denoting the communication-wise out- and in-neighbors of 
$\Sigma_i^{DG}$, respectively. The distributed global 
controller is given by:
\begin{equation}\label{ControllerG}
    u_{iG}(t) \!\triangleq\! \sum_{j\in\bar{\mathcal{F}}_i^-} 
    k_{ij}^c\left(\frac{I_{ti}(t)}{P_{ni}} \!-\! 
    \frac{I_{tj}(t)}{P_{nj}}\right) \!=\! 
    \sum_{j\in\mathbb{N}_N} K_{ij} x_j(t),
\end{equation}
where $k_{ij}^c \in \mathbb{R}$ is the consensus controller 
gain, and $K_{ij}$ is the distributed consensus controller gain matrix defined as:
\begin{equation}\label{k_ij}
    K_{ij} \triangleq
    \frac{1}{L_{ti}}
    \begin{bmatrix}
        0 & 0 & 0 \\
        0 & [K_I]_{ij} & 0 \\
        0 & 0 & 0
    \end{bmatrix}, \quad \forall i,j \in \mathbb{N}_N,
\end{equation}
where $[K_I]_{ij}$ is the $(i,j)$-th element of $K_I \in \mathbb{R}^{N\times N}$, defined as:
\begin{equation}
    [K_I]_{ij} \triangleq
    \begin{cases}
        \dfrac{\sum_{j\in\mathcal{F}_i^-} k_{ij}^c}{P_{ni}} 
        & \text{if } i = j, \\[10pt]
        \dfrac{-k_{ij}^c}{P_{nj}} 
        & \text{if } i \neq j \text{ and } 
        j \in \bar{\mathcal{F}}_i^-, \\[10pt]
        0 & \text{otherwise.}
    \end{cases}
\end{equation}

Note that only the $(2,2)$-th element in each block 
$K_{ij}$ is non-zero by construction. The matrix 
$K_I \in \mathbb{R}^{N \times N}$ satisfies the 
weighted Laplacian property:
\begin{equation}
    K_I P_n \mathbf{1}_N = 0,
\end{equation}
where $P_n = \text{diag}([P_{ni}]_{i\in\mathbb{N}_N})$ 
and $\mathbf{1}_N \in \mathbb{R}^N$ is the vector of 
ones. This ensures that the distributed control vanishes 
when proportional current sharing is achieved among all 
DGs.



Finally, the overall control input $u_i(t)$ applied to the VSC of $\Sigma_i^{DG}$ (i.e., as $V_{ti}(t)$ in \eqref{DGEQ}) can be expressed as
\begin{equation}\label{controlinput}
    u_i(t) \triangleq V_{ti}(t) =  u_{iS} + u_{iL}(t) + u_{iG}(t),
\end{equation}
where $u_{iL}$ is given by \eqref{Controller}, $u_{iG}$ is given by \eqref{ControllerG} and $u_{iS}$ represents the steady-state control input.

As we will see in the sequel, steady-state control input $u_{iS}$ in \eqref{controlinput} also plays a crucial role in achieving the desired equilibrium point of the DC MG. In particular, this steady-state component ensures that the system maintains its operating point that satisfies both voltage regulation and current-sharing objectives, irrespective of the existing CPS and VSC input saturation nonlinearities. The specific structure and properties of $u_{iS}$ will be characterized through our stability analysis presented in Sec. \ref{Sec:Equ_Analysis}. 



\subsection{Closed-Loop Dynamics of the DC MG}

By combining \eqref{DGEQ} and \eqref{error}, and using 
the dead-zone decomposition $\text{sat}(u_i) = u_i + 
\phi_i(u_i)$, the overall dynamics of 
$\Sigma_i^{DG}$, $i\in\mathbb{N}_N$ can be written as
\begin{subequations}\label{statespacemodel}
\begin{align}
       \frac{dV_i}{dt}& =\frac{1}{C_{ti}}\big(I_{ti} \!-\! I_{Li}(V_i) \!-\! I_i \!+\! w_{vi}(t)\big), \label{Eq:ss:voltages}\\
        \frac{dI_{ti}}{dt}& =\frac{1}{L_{ti}}\big(\!-\!V_i \!-\! R_{ti}I_{ti} \!+\! u_i \!+\! \phi_i(u_i) \!+\! w_{ci}(t)\big), \label{Eq:ss:currents}\\
        \frac{dv_i}{dt}& = V_i \!-\! V_{ri} \!-\! K_{aw,i}\phi_i(u_i) \label{Eq:ss:ints},
\end{align}
\end{subequations}
where the terms  $I_i$, $I_{Li}(V_i)$, and $u_i$ can be substituted from \eqref{Eq:DGCurrentNetOut}, \eqref{Eq:LoadModel}, and \eqref{controlinput}, respectively. We can restate \eqref{statespacemodel} as
\begin{equation}\label{Eq:DGCompact}
\begin{aligned}
\dot{x}_i(t) &= A_ix_i(t) + B_iu_i(t) + B_i^{aw}\phi_i(u_i(t)) \\
&+ E_id_i(t) + \xi_i(t) + g_i(x_i(t)),
\end{aligned}
\end{equation}
where $x_i(t)$ is the DG state as defined in \eqref{Eq:DGstate}, $\phi_i(u_i(t))$ is the dead-zone nonlinearity defined in \eqref{Eq:Saturation} applied to the VSC input $u_i(t)=V_{ti}$, $d_i(t)$ is the exogenous input (disturbance) defined as:
\begin{equation}
d_i(t) \triangleq 
\bar{w}_i + w_i(t),
\end{equation}
with 
$\bar{w}_i \triangleq \bm{
    -\Bar{I}_{Li}  & 0 & -V_{ri}
}^\T$ representing the fixed (mean) known disturbance and
$w_i(t) \triangleq \bm{w_{vi}(t) & w_{ci}(t) & 0}^\T$ representing the zero-mean unknown disturbance, $\xi_i(t)$ is the transmission line coupling input defined as 
$\xi_i(t) \triangleq \begin{bmatrix}
    -C_{ti}^{-1}\sum_{l\in \mathcal{E}_i} \mathcal{B}_{il}I_l(t) & 0 & 0
\end{bmatrix}^\top$, $g_i(x_i(t))$ represents the nonlinear vector field due to the CPL defined as 
$$
g_i(x_i(t)) \triangleq C_{ti}^{-1}\begin{bmatrix} -\frac{P_{Li}}{V_i} & 0 & 0 \end{bmatrix}^\T,
$$
and $A_i$ and $B_i$ are system matrices respectively defined as 
\begin{equation}\label{Eq:DG_Matrix_definition}
A_i \!\triangleq\! 
\begin{bmatrix}
  -\frac{Y_{Li}}{C_{ti}} & \frac{1}{C_{ti}} & 0\\
-\frac{1}{L_{ti}} & -\frac{R_{ti}}{L_{ti}} & 0 \\
1 & 0 & 0
\end{bmatrix}, 
B_i \!\triangleq\!
\begin{bmatrix}
 0 \\ \frac{1}{L_{ti}} \\ 0
\end{bmatrix}, 
B_i^{aw} \!\triangleq\! \begin{bmatrix} 0 \\ \frac{1}{L_{ti}} \\ -K_{aw,i} \end{bmatrix},
\end{equation}
and $E_i \triangleq \diag\left(C_{ti}^{-1}, L_{ti}^{-1}, 1\right)$ is the disturbance input matrix.

Similarly, using \eqref{line}, the state space representation of the transmission line $\Sigma_l^{Line}$ can be written in a compact form:
\begin{equation}\label{Eq:LineCompact}
    \dot{\bar{x}}_l(t) = \bar{A}_l\bar{x}_l(t) + \bar{B}_l\bar{u}_l(t) + \bar{E}_l\bar{w}_l(t),
 \end{equation}
where $\bar{x}_l \triangleq I_l$ is the transmission line state,  $\bar{E}_l \triangleq \bm{\frac{1}{L_{l}}}$ is the disturbance matrix, and $\bar{A}_l$ and $\Bar{B}_l$ are the system matrices respectively defined as
\begin{equation}\label{Eq:LineMatrices}
\bar{A}_l \triangleq 
\begin{bmatrix}
-\frac{R_l}{L_l}
\end{bmatrix}\ 
\mbox{ and }\ 
\Bar{B}_l \triangleq  \begin{bmatrix}
\frac{1}{L_l}
\end{bmatrix}.  
\end{equation}

\subsection{Networked System Model}\label{Networked System Model}

Let us define $u\triangleq[u_i]_{i\in\N_N}^\T$ and $\Bar{u}\triangleq[\Bar{u}_l]_{l\in\N_L}^\T$ respectively as vectorized control inputs of DGs and lines, 
$x\triangleq[x_i^\T]_{i\in\N_N}^\T$ and $\Bar{x}\triangleq[\Bar{x}_l]_{l\in\N_L}^\T$ respectively as the full states of DGs and lines, $w\triangleq[w_i^\T]_{i\in\N_N}^\T$ and $\bar{w}\triangleq[\bar{w}_l]_{l\in\N_L}^\T$ respectively as disturbance inputs of DGs and lines.

Using these notations, we can now represent the closed-loop DC MG as two sets of subsystems (i.e., DGs and lines) interconnected with disturbance inputs through a static interconnection matrix $M$ as shown in Fig. \ref{netwoked}. From comparing Fig. \ref{netwoked} with Fig. \ref{Networked}, it is clear that the DC MG takes a similar form to a standard networked system discussed in Sec. \ref{SubSec:NetworkedSystemsPreliminaries}.

\begin{figure}
    \centering
    \includegraphics[width=0.9\columnwidth]{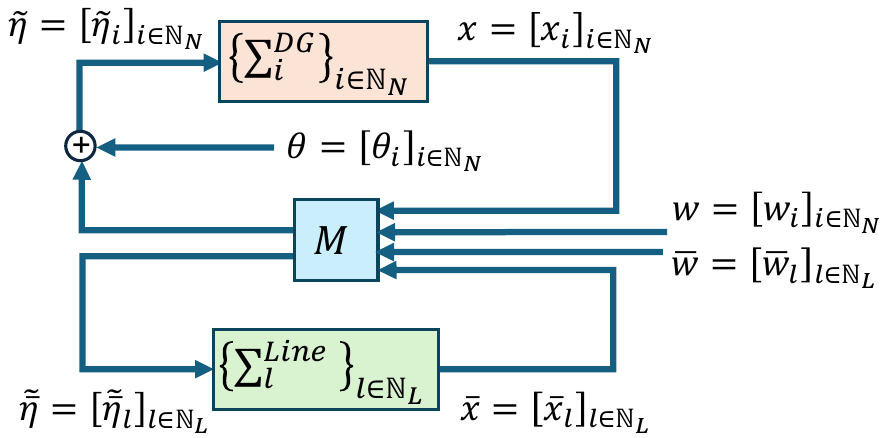}
    \caption{DC MG dynamics as a networked system configuration.}
    \label{netwoked}
\end{figure}

To identify the specific structure of the interconnection matrix $M$ in Fig. \ref{netwoked} (i.e., for DC MG), we need to closely observe how the inputs and outputs of DGs and lines are interconnected, how they are coupled with disturbance inputs, and how they define performance goals (outputs) of the DC MG.

To this end, we first use \eqref{Eq:DGCompact} and \eqref{controlinput} to state the closed-loop dynamics of $\Sigma_i^{DG}$ as:
\begin{equation}\label{closedloopdynamic}
    \dot{x}_i(t) \!=\! (A_i+B_iK_{i0})x_i(t) \!+\! g_i(x_i(t)) \!+\! B_i^{aw}\phi_i(u_i(t)) \!+\! \eta_i(t),
\end{equation}
where $\eta_i$ is defined as 
\begin{equation}
\label{Eq:DGClosedLoopDynamics_varphi}
\eta_i(t) \triangleq E_iw_i(t) +\sum_{l\in\mathcal{E}_i}\Bar{C}_{il}\Bar{x}_l(t)+\sum_{j\in\bar{\mathcal{F}}_i^-}K_{ij}x_j(t) + \theta_i,
\end{equation}
with
$\Bar{C}_{il} \triangleq -C_{ti}^{-1}\begin{bmatrix}
\mathcal{B}_{il} & 0 & 0
\end{bmatrix}^\top$, 
$$\theta_i \triangleq  E_i\bar{w}_i + B_iu_{iS} - B_ik_{i0}^P V_{ri},$$ where $K_{ij}$ is as defined in \eqref{k_ij}.

By vectorizing \eqref{Eq:DGClosedLoopDynamics_varphi} over all $i\in\N_N$, we get 
\begin{equation}\label{Eq:DGClosedLoopInputVector}
    \eta \triangleq  Ew+\Bar{C}\Bar{x}+Kx+\theta,
\end{equation}
where $\eta \triangleq [\eta_i^\T]_{i\in \N_N}^\T$ represents the effective input vector to the DGs (see Fig. \ref{netwoked}), $E \triangleq \diag([E_i]_{i\in\N_N})$ represents the disturbance matrix of DGs, $\Bar{C}\triangleq[\Bar{C}_{il}]_{i\in\mathbb{N}_N,l\in\mathbb{N}_L}$, $K \triangleq [K_{ij}]_{i,j\in\mathbb{N}_N}$, and $\theta\triangleq [\theta_i]_{i\in\N_N}^\T$ represents a constant (time-invariant) input vector applied to the DGs.

\begin{remark}
The block matrices $K$ and $\Bar{C}$ in \eqref{Eq:DGClosedLoopInputVector} are indicative of the communication and physical topologies of the DC MG, respectively. In particular, the $(i,j)$\tsup{th} block in $K$, i.e., $K_{ij}$, indicates a communication link from $\Sigma_j^{DG}$ to $\Sigma_i^{DG}$. Similarly, $(i,l)$\tsup{th} block in $\bar{C}$, i.e., $\bar{C}_{i,l}$ indicates a physical link between $\Sigma_i^{DG}$ and $\Sigma_l^{Line}$.
\end{remark}


Similarly to DGs, using \eqref{Eq:LineCompact}, we state the closed-loop dynamics of $\Sigma_l^{Line}$ as  
\begin{equation}
    \dot{\bar{x}}_l(t) = \bar{A}_l\bar{x}_l(t) + \bar{\eta}_l(t),
\end{equation}
where $\bar{\eta}_l(t)$ is defined as 
\begin{equation}\label{Eq:linecontroller}
    \Bar{\eta}_l(t) \triangleq \sum_{i\in\mathcal{E}_l}C_{il}x_i(t) + \bar{E}_l\bar{w}_l(t),
\end{equation}
with $C_{il}\triangleq \begin{bmatrix}
    \mathcal{B}_{il} & 0 & 0
\end{bmatrix}$ (note also that $C_{il} = -C_{ti}\bar{C}_{il}^\T$). By vectorizing \eqref{Eq:linecontroller} over all $l\in\N_L$, we get:
\begin{equation}\label{ubar}
    \Bar{\eta} \triangleq Cx + \bar{E}\bar{w},
\end{equation}
where $\bar{\eta} \triangleq [\bar{\eta}_l]_{l\in\N_L}^\T$ represents the effective input vector to the lines, $C\triangleq[C_{il}]_{l\in\mathbb{N}_L,i\in\mathbb{N}_N}$ (note also that $C = - \bar{C}^\T C_t$ where $C_t \triangleq \diag([C_{ti}\I_{3}]_{i\in\N_N})$), $\bar{E} \triangleq \diag([\bar{E}_l]_{l\in\N_L})$, and $\bar{w} \triangleq [\bar{w}_l]_{l\in\N_L}$.





Finally, using \eqref{Eq:DGClosedLoopInputVector} and \eqref{ubar}, we can identify the interconnection relationship:  
\begin{equation}
\begin{bmatrix} \eta \\ \bar{\eta} \end{bmatrix}
= M
\begin{bmatrix} x \\ \bar{x} \\ w \\ \bar{w} \end{bmatrix}
+ \begin{bmatrix} \theta \\ \0 \end{bmatrix},
\end{equation}
where the interconnection matrix $M$ takes the form:
\begin{equation}\label{Eq:MMatrix}
M \triangleq 
\begin{bmatrix}
    K & \Bar{C} & E & \0 \\
    C & \0 & \0 & \bar{E}\\
\end{bmatrix}.
\end{equation}

When the physical topology $\mathcal{G}^p$ is predefined, so are the block matrices $\Bar{C}$ and $C$ (recall $C = -\bar{C}^\T C_t$). This leaves only the block matrix $K$ within the block matrix $M$ as a tunable quantity for optimizing the desired properties of the closed-loop DC MG system. Note that synthesizing $K$ simultaneously determines the distributed global controllers \eqref{ControllerG} and the communication topology $\mathcal{G}^c$. In the following two sections, we formulate the networked system's error dynamics (around a desired operating point) and provide a systematic dissipativity-based approach to synthesize the block matrix $K$ that enforces dissipativity of the closed-loop error dynamics (from disturbance inputs to a given performance output).

\section{Nonlinear Networked Error Dynamics}\label{Sec:ControlDesign}

This section establishes the mathematical foundation for the stability and robustness analysis of the DC MG with CPL and VSC input saturation nonlinearities. First, a rigorous equilibrium point analysis is conducted to characterize steady-state behavior and derive necessary conditions for simultaneous voltage regulation and current sharing. Next, error dynamics are developed around the identified equilibrium point, explicitly accounting for CPL and VSC input saturation nonlinearities. These error dynamics are then cast into a standard networked system structure with clearly defined performance outputs and disturbance inputs, providing a complete state-space representation of both DG and line error subsystems.

\subsection{Equilibrium Point Analysis of the DC MG}\label{Sec:Equ_Analysis}

In this section, we analyze the equilibrium conditions of the DC MG to establish mathematical relationships between system parameters and steady-state (equilibrium) behavior. This analysis is crucial for identifying the necessary conditions for achieving both voltage regulation and proportional current sharing simultaneously. We pay particular attention to the impact of CPL and VSC saturation nonlinearities, which potentially can lead to instability in the DC MG system.

\begin{lemma}\label{Lm:equilibrium}
Assuming all zero mean unknown disturbance components to be zero, i.e., $w_i(t)=\0, \forall i\in\N_N$ and $\bar{w}_l(t)=0,\forall l\in\N_L$, for a given reference voltage vector $V_r$, under a fixed control input $u(t) = u_E \triangleq[u_{iE}]_{i\in\N_N}$ defined as
\begin{equation}
    u_E \triangleq [\I + R_t(\mathcal{B}R^{-1}\mathcal{B}^\T + Y_L)]V_r + R_t(\bar{I}_L + \diag(V_r)^{-1}P_L), 
\end{equation}
there exists an equilibrium point for the DC MG characterized by reference voltage vector $V_r \triangleq [V_{ri}]_{i\in\N_N}^\T$, constant current load vector $\bar{I}_L \triangleq [\bar{I}_{Li}]_{i\in\N_N}$, and CPL vector $P_L \triangleq [P_{Li}]_{i\in\N_N}^\T$, given by:
\begin{equation}\label{Eq:Equilibrium}
\begin{aligned}
V_E &= V_{r},\\
I_{tE} &= (\mathcal{B}R^{-1}\mathcal{B}^\T + Y_L)V_r + \bar{I}_L + \diag(V_r)^{-1}P_L,\\
\bar{I}_E &= R^{-1}\mathcal{B}^\T V_r,
\end{aligned}
\end{equation}
provided that the equilibrium control input $u_E$ satisfies the VSC saturation constraints:
\begin{equation*}
    V_{ti}^{\min} \leq u_{iE} \leq V_{ti}^{\max}, \ \forall i \in\N_N,
\end{equation*}
where we define the state equilibrium vectors $V_E\triangleq[V_{iE}]_{i\in\N_N}$,
$I_{tE}\triangleq[I_{tiE}]_{i\in\N_N}$,
$\bar{I}_E\triangleq[\bar{I}_{lE}]_{l\in\N_L}$, 
and the system parameters 
$Y_L\triangleq\diag([Y_{Li}]_{i\in\N_N})$,
$R_t\triangleq\diag([R_{ti}]_{i\in\N_N})$,
and $R\triangleq\diag([R_l]_{l\in\mathbb{N}_L})$.
\end{lemma}

\begin{proof}
The equilibrium state of the closed-loop dynamic $\Sigma_i^{DG}$ \eqref{Eq:DGCompact} satisfies:
\begin{equation}\label{Eq:EqDGCompact}
A_ix_{iE}+B_iu_{iE}+g_i(x_{iE})+E_id_{iE}+\xi_{iE} = 0,
\end{equation}
where $x_{iE} \triangleq \bm{V_{iE} & I_{tiE} & v_{iE}}^\T$ represents the equilibrium state components of DG, and $w_{iE}\triangleq \bar{w}_i$ and $\xi_{iE}$ represent the equilibrium values of disturbance and interconnection terms, respectively. Thus, we get
\begin{equation}
\begin{aligned}
&\begin{bmatrix} -\frac{Y_{Li}}{C_{ti}} & \frac{1}{C_{ti}} & 0\\ 
-\frac{1}{L_{ti}} & -\frac{R_{ti}}{L_{ti}} & 0 \\ 
1 & 0 & 0\end{bmatrix} 
\begin{bmatrix} V_{iE} \\ I_{tiE} \\ v_{iE} \end{bmatrix} 
+\begin{bmatrix} 0 \\ \frac{1}{L_{ti}} \\ 0 \end{bmatrix}u_{iE}\\
&+\begin{bmatrix} -\frac{P_{Li}}{C_{ti}V_{iE}} \\ 0 \\ 0 \end{bmatrix}
+ E_i\bar{w}_i + \xi_{iE} = 0.
\end{aligned}
\end{equation}

From the respective rows of this matrix equation, we get:
\begin{align}
\label{Eq:steadystatevoltage}
&\frac{1}{C_{ti}}\left(-Y_{Li}V_{iE} \!+\! I_{tiE} \!-\! \bar{I}_{Li} \!-\! \frac{P_{Li}}{V_{iE}} \!-\! \sum_{l\in \mathcal{E}_i} \mathcal{B}_{il}\bar{I}_{lE}\right) \!=\! 0,\\
\label{Eq:steadystatecurrent}
&\frac{1}{L_{ti}}\left(-V_{iE} - R_{ti}I_{tiE} + u_{iE}\right) = 0, \\
\label{Eq:steadystateintegrator}
&V_{iE} - V_{ri} = 0.
\end{align}
To simplify the first equation further, we need to know an expression for $\bar{I}_{lE}$. However, from the last two equations above, we can obtain
\begin{align}
 \label{Eq:voltageEquil}
    &V_{iE} = V_{ri},\\  
    \label{Eq:controlEquil}
    &u_{iE} = V_{iE} + R_{ti}I_{tiE}.
\end{align}

Note that the equilibrium state of the $\Sigma_l^{Line}$ \eqref{Eq:LineCompact} satisfies:
\begin{equation}\label{Eq:EqLineCompact}
    \bar{A}_l\bar{x}_{lE} + \bar{B}_l\bar{u}_{lE} = 0,
 \end{equation}
where $\bar{x}_{lE} \triangleq \bar{I}_{lE}$ represents the equilibrium state of line, and $\bar{u}_{lE}\triangleq \sum_{i\in \mathcal{E}_l}\mathcal{B}_{il}V_{iE}$ represent the equilibrium values of control input. Therefore, we get:
\begin{equation}\label{Eq:steadystateLine}
    -\frac{R_l}{L_l}\bar{I}_{lE} + \frac{1}{L_l}\sum_{i\in \mathcal{E}_l}\mathcal{B}_{il}V_{iE} = 0,
\end{equation}
leading to
\begin{equation}
\label{Eq:LineEquil}
    \bar{I}_{lE} = \frac{1}{R_l}\sum_{i\in \mathcal{E}_l}\mathcal{B}_{il}V_{iE} = \frac{1}{R_l}\sum_{j\in \mathcal{E}_l}\mathcal{B}_{jl}V_{jE},
\end{equation}
which can be applied in \eqref{Eq:steadystatevoltage} (together with \eqref{Eq:voltageEquil}) to obtain
\begin{equation}\label{Eq:currentEquil}
    -Y_{Li}V_{ri} + I_{tiE} - \sum_{l\in\mathcal{E}_i}\mathcal{B}_{il}\bigg(\frac{1}{R_l}\sum_{j\in \mathcal{E}_l}\mathcal{B}_{jl}V_{rj}\bigg)- \bar{I}_{Li} - \frac{P_{Li}}{V_{ri}} = 0.
\end{equation}

We next vectorize these equilibrium conditions. From \eqref{Eq:voltageEquil}, we get $V_E=V_r$. From \eqref{Eq:controlEquil}, i.e., the control equilibrium equation, we get:
\begin{equation}\label{Eq:controlEquilVec}
    u_E = V_r + R_tI_{tE}.
\end{equation}
Vectorizing the voltage dynamics equation \eqref{Eq:currentEquil}, we get:
\begin{equation}\nonumber
    -Y_L V_r + I_{tE} - \mathcal{B}R^{-1}\mathcal{B}^\T V_r - \bar{I}_L - \diag(V_r)^{-1}P_L = 0,
\end{equation}
leading to
\begin{equation}\nonumber
    I_{tE} = (\mathcal{B}R^{-1}\mathcal{B}^\T + Y_L)V_r + \bar{I}_L + \diag(V_r)^{-1}P_L.
\end{equation}
Therefore, the vectorized control equilibrium equation can be expressed as:
\begin{equation}
\begin{aligned}\nonumber
    u_E &= V_r + R_t I_{tE} \\
    &= V_r + R_t((\mathcal{B}R^{-1}\mathcal{B}^\T + Y_L)V_r + \bar{I}_L + \diag(V_r)^{-1}P_L)\\
    &= [\I + R_t(\mathcal{B}R^{-1}\mathcal{B}^\T + Y_L)]V_r + R_t(\bar{I}_L + \diag(V_r)^{-1}P_L).
    \end{aligned}
\end{equation}
Finally, from \eqref{Eq:LineEquil}, i.e., the equilibrium line currents, we get 
\begin{equation}\nonumber
    \bar{I}_E = R^{-1}\mathcal{B}^\T V_r,
\end{equation}
which completes the proof, as we have derived all the required equilibrium conditions. 



\end{proof}


\begin{remark}
The uniqueness of the equilibrium point established in 
Lm. \ref{Lm:equilibrium} is mathematically guaranteed 
under the condition that $V_{ri} > 0$, $\forall i \in 
\mathbb{N}_N$, which is physically reasonable since all 
DG voltages must be positive in a DC MG. Specifically, 
the diagonal matrices $R$, $R_t$, and $Y_L$ have strictly 
positive elements, making them positive definite, while 
the incidence matrix $\mathcal{B}$ maintains full rank by 
virtue of the connected network topology. Consequently, 
the coefficient matrix $(\mathcal{B}R^{-1}\mathcal{B}^\top 
+ Y_L)$ in \eqref{Eq:Equilibrium} is invertible, ensuring 
a unique one-to-one mapping from any given reference 
voltage vector $V_r$ to all equilibrium variables under 
specified loading conditions. Note that the CPL term 
$\text{diag}(V_r)^{-1}P_L$ in \eqref{Eq:Equilibrium} 
is well-defined and finite under this condition since 
$V_{ri} > 0$, $\forall i \in \mathbb{N}_N$, and hence 
does not affect the uniqueness of the equilibrium point.
\end{remark}


\begin{remark}\label{Lm:currentsharing}
At the equilibrium, we require the condition for proportional current sharing among DGs to meet (i.e., \eqref{Eq:PropCurrSharing}), and thus, we require
\begin{equation}\label{Eq:currentsharing}
\frac{I_{tiE}}{P_{ni}} = I_s \Longleftrightarrow I_{tiE} = P_{ni}I_s,\ 0 \leq I_s \leq 1,\ \forall i\in\mathbb{N}_N, 
\end{equation}
which can be expressed in vectorized form as $I_{tE} = P_{n} \mathbf{1}_N I_s$, where $P_n\triangleq\diag([P_{ni}]_{i\in\mathbb{N}_N})$. Using this requirement in \eqref{Eq:controlEquilVec}, we get $u_E = V_r +  R_t P_{n} \mathbf{1}_N I_s,$ i.e., 
$$
u_{iE} = V_{ri} +  R_{ti} P_{ni} I_s, \quad \forall i\in\N_N. 
$$
Therefore, to achieve this particular control equilibrium (which satisfies both voltage regulation and current sharing objectives), we need to select our steady-state control input in \eqref{controlinput} as:
\begin{equation}\label{Eq:Rm:SteadyStateControl}
    u_{iS} = V_{ri} + R_{ti}P_{ni}I_s, \quad \forall i\in\N_N. 
\end{equation}
This is because at the equilibrium point, local control $u_{iL}$ \eqref{Controller} and distributed global control $u_{iG}$ \eqref{ControllerG} components are, by definition, zero for any $i\in\N_N$. Furthermore, to ensure the equilibrium control input $u_{iE} = u_{iS}$ does not violate the VSC saturation constraint, we require:
\begin{equation*}
    V_{ti}^{\min} \leq V_{ri} + R_{ti}P_{ni} I_s \leq V_{ti}^{\max}, \ \forall i\in\N_N.
\end{equation*}
\end{remark}

In conclusion, using Lm. \ref{Lm:equilibrium} and Rm. \ref{Lm:currentsharing}, for the equilibrium of DC MG to satisfy the voltage regulation and current sharing conditions, we require:
\begin{equation}
    \begin{aligned}
        u_E &= [\I + R_t(\mathcal{B}R^{-1}\mathcal{B}^\T + Y_L)]V_r + R_t(\bar{I}_L + \diag(V_r)^{-1}P_L)  \\
        &= V_r +  R_t P_{n} \mathbf{1}_N I_s = u_S\\
        V_E &= V_r, \\
        I_{tE} &= (\mathcal{B}R^{-1}\mathcal{B}^\T + Y_L)V_r + \bar{I}_L + \diag(V_r)^{-1}P_L\\
        &= P_n \textbf{1}_N I_s,\\
        \bar{I}_E &= R^{-1} \mathcal{B}^\T V_r.
    \end{aligned}
\end{equation}

The following theorem formalizes the optimization problem derived from Lm. \ref{Lm:equilibrium} and Rm. \ref{Lm:currentsharing}, for the selection of $V_r$ and $I_s$, that ensures the existence of an equilibrium state that satisfies voltage regulation and current sharing conditions while also respecting reference voltage limits $V_{\min}$ and $V_{\max}$ and the current sharing coefficient $I_s \in [0,1]$. A formal proof is omitted as the result follows directly from the previous discussion.







\begin{theorem}\label{Th:Equilibrium}
To ensure the existence of an equilibrium point that satisfies the voltage regulation and current sharing objectives, the reference voltages $V_r$ and current sharing coefficient $I_s$ should be a feasible solution in the optimization problem:
\begin{equation}\label{Eq:Th:Equilibrium}
\begin{aligned}
\min_{V_r,I_s}\ &\alpha_V\Vert V_r - \bar{V}_r\Vert^2 + \alpha_I I_s \\
\mbox{Sub. to:}\ & V_{\min} \leq V_r \leq V_{\max},
\quad 0 \leq I_s \leq 1,
\end{aligned}
\end{equation}
$$P_n \mathbf{1}_N I_s - (\mathcal{B}R^{-1}\mathcal{B}^\T + Y_L)V_r = \bar{I}_L + \diag(V_r)^{-1}P_L,$$
$$
V_t^{\min} \leq V_r + R_tP_n \mathbf{1}_N I_s \leq V_t^{\max}, 
$$
where $\bar{V}_r$ is a desired reference voltage value, $V_t^{\min} \triangleq [V_{ti}^{\min}]_{i\in\N_N}$, $V_t^{\max} \triangleq [V_{ti}^{\max}]_{i\in\N_N}$, and $\alpha_V > 0$ and $\alpha_I > 0$ are two normalizing cost coefficients.
\end{theorem}

It is worth noting that the above optimization problem becomes an LMI problem (convex) when the CPL is omitted (i.e., when $P_L=\0$). 
This formulation ensures proper system operation through multiple aspects. The equality constraint ensures that the current-sharing objective is satisfied across all DG units. The reference-voltage bounds maintain system operation within safe and efficient limits via inequality constraints on $V_r$. Furthermore, the constraint on $I_s$ ensures that the current-sharing coefficient remains properly normalized in practical implementations. 
As stated earlier, the nonlinear term $\text{diag}(V_r)^{-1}P_L$ renders the optimization problem \eqref{Eq:Th:Equilibrium} non-convex, which can be addressed using standard nonlinear programming solvers \cite{najafirad2025dissipativity}.




\subsection{Nonlinear Error Dynamics with CPL}
The network system representation described in Sec. \ref{Networked System Model} can be simplified by considering the error dynamics around the identified equilibrium point in Lm. \ref{Lm:equilibrium}. As we will see in the sequel, the resulting error dynamics can be seen as a networked system (called the networked error system) comprised of DG error subsystems, line error subsystems, external disturbance inputs, and performance outputs.


We first define error variables that capture deviations from the identified equilibrium:
\begin{subequations}
\begin{align}
    \tilde{V}_i &= V_i - V_{iE} = V_i - V_{ri}, \label{Eq:errorvariable1}\\
    \tilde{I}_{ti} &= I_{ti} - I_{tiE} = {I_{ti} - P_{ni}I_s}, \label{Eq:errorvariable2}\\
     \tilde{v}_i &= v_i - v_{iE}, \label{Eq:errorvariable3}\\
    \tilde{I}_l &= I_{l} - \bar{I}_{lE} = I_l - \frac{1}{R_l}\sum_{i\in\E_l}\mathcal{B}_{il}V_{ri}.
\end{align}
\end{subequations}
Now, considering the dynamics \eqref{Eq:ss:voltages}-\eqref{Eq:ss:ints}, the equilibrium point established in Lm. \ref{Lm:equilibrium}, and the proposed a hierarchical control strategy $u_i(t)$ \eqref{controlinput}, the error dynamics can then be derived as follows.

The voltage error dynamics can be derived using \eqref{Eq:ss:voltages} and \eqref{Eq:errorvariable1} as:
\begin{equation}\label{Eq:currenterrordynamic}
\begin{aligned}
    \dot{\tilde{V}}_i = &-\frac{Y_{Li}}{C_{ti}}(\tilde{V}_i+V_{ri}) + \frac{1}{C_{ti}}(\tilde{I}_{ti} +P_{ni}I_s) - \frac{1}{C_{ti}}\bar{I}_{Li}\\
    &- \frac{1}{C_{ti}}\sum_{l\in \mathcal{E}i} \mathcal{B}_{il}(\tilde{I}_l + \frac{1}{R_l}\sum_{j\in\E_l}\mathcal{B}_{jl}V_{rj})\\
    &- \frac{1}{C_{ti}}(\tilde{V}_i + V_{ri})^{-1}P_{Li}+ \frac{1}{C_{ti}}w_{vi}\\
    \equiv& \frac{1}{C_{ti}}\Big(\varphi_V + \psi_V +g_i(\tilde{V}_i)\Big) + \frac{1}{C_{ti}}w_{vi},
\end{aligned}
\end{equation}
where 
\begin{subequations}
\begin{align}
   \varphi_V &\triangleq -Y_{Li}\tilde{V}_i + \tilde{I}_{ti} - \sum_{l\in \mathcal{E}i} \mathcal{B}_{il}\tilde{I}_l,\\
    \psi_V &\!\triangleq\! -Y_{Li}V_{ri}+P_{ni}I_s \!-\! \bar{I}_{Li} \!-\! \sum_{l\in \mathcal{E}i} \frac{\mathcal{B}_{il}}{R_l}\sum_{j\in\E_l}\mathcal{B}_{jl}V_{rj} \!-\! \frac{V_{ri}}{P_{Li}}, \label{Eqvoltageerror}\\
    g_i(\tilde{V}_i) &\triangleq V_{ri}^{-1}P_{Li} - (\tilde{V}_i + V_{ri})^{-1}P_{Li}.
\end{align}
\end{subequations}

The current error dynamics can be obtained using \eqref{Eq:ss:currents} and \eqref{Eq:errorvariable2} as:
\begin{equation}\label{Eq:voltageerrordynamic}
\begin{aligned}
    &\dot{\tilde{I}}_{ti} = -\frac{1}{L_{ti}}(\tilde{V}_i+V_{ri}) 
    - \frac{R_{ti}}{L_{ti}}(\tilde{I}_{ti} +P_{ni}I_s) \!+\! \frac{1}{L_{ti}}w_{ci} \\
    &+\frac{1}{L_{ti}}\big(u_{iS} \!+\! k_{i0}^P\tilde{V}_i 
    \!+\! k_{io}^I\tilde{v}_i \!+\! 
    \sum_{j\in\bar{\mathcal{F}}_i^-} k_{ij}^c
    \big(\frac{\tilde{I}_{ti}}{P_{ni}} \!-\! 
    \frac{\tilde{I}_{tj}}{P_{nj}}\big)
    \!+\! \phi_i(u_i)\big) ,\\
    &\equiv \frac{1}{L_{ti}}\Big(\varphi_I + \psi_I\Big) 
    + \frac{1}{L_{ti}}\phi_i(u_i)
    + \frac{1}{L_{ti}}w_{ci},
\end{aligned}
\end{equation}
where
\begin{subequations}
\begin{align}
    \varphi_I \triangleq& \!-\! \tilde{V}_i \!-\! R_{ti}\tilde{I}_{ti} \!+\! k_{io}^P\tilde{V}_i \!+\! k_{io}^I\tilde{v}_i
    \!+\! \sum_{j\in\bar{\mathcal{F}}_i^-} k_{ij}(\frac{\tilde{I}_{ti}}{P_{ni}} \!-\! \frac{\tilde{I}_{tj}}{P_{nj}})),\\
     \psi_I \triangleq& -V_{ri}-R_{ti}P_{ni}I_s + u_{iS}. \label{Eqcurrenteerror}
\end{align}
\end{subequations}

The integral error dynamics can be achieved by using \eqref{Eq:ss:ints} and \eqref{Eq:errorvariable3} as:
\begin{equation}\label{Eq:integralerrordynamic}
    \dot{\tilde{v}}_i = \tilde{V}_i -K_{aw,i}\phi_i(u_i),
\end{equation}
where $\phi_i(u_i)$ denotes the dead-zone of the scalar VSC command $u_i \in \mathbb{R}$ \eqref{Controller}.


It is worth noting that, as a consequence of the equilibrium analysis and the steady state control input selection (see \eqref{Eq:currentsharing} and \eqref{Eq:Rm:SteadyStateControl}), the terms $\psi_V$ \eqref{Eqvoltageerror} and $\psi_I$ \eqref{Eqcurrenteerror} get canceled out. Therefore, for each DG error subsystem $\tilde{\Sigma}_i^{DG}, i\in\N_N$, we have an error state vector $\tilde{x}_i = \begin{bmatrix} \tilde{V}_i & \tilde{I}_{ti} & \tilde{v}_i \end{bmatrix}^\T$ with the dynamics: 
\begin{equation}\label{Eq:DG_error_dynamic}
    \dot{\tilde{x}}_i = (A_i + B_i K_{i0})\tilde{x}_i  + g_i(\tilde{x}_i) + B_i^{aw}\phi_i(u_i) + \tilde{\eta}_i, 
\end{equation}
where $\tilde{\eta}_i$ represents the interconnection input defined as
\begin{equation}
\tilde{\eta}_i = \sum_{l\in\mathcal{E}_i}\Bar{C}_{il}\tilde{\Bar{x}}_l +  \sum_{j\in\bar{\mathcal{F}}_i^-}K_{ij}\tilde{x}_j + E_iw_i,
\end{equation}
and $g_i(\tilde{x}_i)$ is the nonlinear vector due to the CPL, defined as
\begin{equation}\label{nonlinear_vector}
    g_i(\tilde{x}_i) = \bm{\frac{1}{C_{ti}}\big(V_{ri}^{-1}P_{Li} - (\tilde{V}_i + V_{ri})^{-1}P_{Li}\big) \\ 0 \\0
    },
\end{equation}
with $A_i$ and $B_i$ as defined in \eqref{Eq:DG_Matrix_definition}.


Following similar steps, we can obtain the dynamics of the transmission line error subsystem $\tilde{\Sigma}_l^{Line}, l\in\N_L$ as:
\begin{equation}\label{Eq:Line_error_dynamic}
    \dot{\tilde{\bar{x}}}_l = \bar{A}_l\tilde{\bar{x}}_l + \tilde{\bar{u}}_l,  
\end{equation}
where $\tilde{\bar{u}}_l$ represents the line interconnection input influenced by DG voltages and disturbances:
\begin{equation}
    \tilde{\bar{u}}_l = \bar{u}_l + \bar{E}_l \bar{w}_l = \sum_{i\in\mathcal{E}_l}B_{il}\tilde{V}_i  + \bar{E}_l \bar{w}_l.
\end{equation}

To ensure robust stability (dissipativity) of this networked error system, we define local performance outputs as follows. For each DG error subsystem $\tilde{\Sigma}_i^{DG}, i\in\N_N$, we define the performance output as:
\begin{equation}
z_i(t) \triangleq H_i\tilde{x}_i(t),
\end{equation}
where $H_i$ can be selected as $H_i \triangleq \I, \forall i\in\N_N$ (not necessarily). 
Similarly, for each line error subsystem $\tilde{\Sigma}_l^{Line}, l\in\N_L$, we define the performance output as:
\begin{equation}
\bar{z}_l(t) \triangleq \bar{H}_l\tilde{\bar{x}}_l(t),
\end{equation}
where $\bar{H}_l$ can be selected as $\bar{H}_l \triangleq \I, \forall l\in\N_L$ (not necessarily). 

Upon vectorizing these performance outputs over all $i\in\N_N$ and $l\in\N_L$ (respectively), we obtain: 
\begin{equation}\label{z}
\begin{aligned}
     z \triangleq H\tilde{x}\quad \mbox{ and } \quad 
     \bar{z} \triangleq \bar{H}\tilde{\bar{x}},
\end{aligned}
\end{equation} 
where $H \triangleq \diag([H_i]_{i\in\N_N})$ and $\bar{H} \triangleq \diag([\bar{H}_l]_{l\in\N_L})$. This choice of performance output mapping provides a direct correspondence between error subsystem states and the performance outputs.

Defining $z \triangleq \bm{z_i^\T}^\T_{i\in\N_N}$ and $\bar{z} \triangleq \bm{\bar{z}_l^\T}^\T_{l\in\N_L}$, we consolidate the performance outputs and disturbance inputs as
\begin{equation}
    \begin{aligned}
        z_c \triangleq \bm{z^\T & \bar{z}^\T}^\T\quad \mbox{ and } \quad 
        w_c \triangleq \bm{w^\T & \bar{w}^\T}^\T.
    \end{aligned}
\end{equation}
The consolidated disturbance vector $w_c$ affects the networked error dynamics, particularly the DG error subsystems and the line error subsystems, respectively, through the consolidated disturbance matrices $E_c$ and $\bar{E}_c$, defined as:
\begin{equation}
    \begin{aligned}
        E_c \triangleq \bm{E & \0}\quad \mbox{ and }\quad 
        \bar{E}_c \triangleq \bm{\0 & \bar{E}}.
    \end{aligned}
\end{equation}

The zero blocks in the $E_c$ and $\bar{E}_c$ indicate that line disturbances do not directly affect DG error subsystem inputs and vice versa. Analogously, the dependence of consolidated performance outputs on the networked error system states can be described using consolidated performance matrices 
\begin{equation}
    \begin{aligned}
        H_c \triangleq \bm{H \\ \0}\quad \mbox{ and }\quad 
        \bar{H}_c \triangleq \bm{\0 \\ \bar{H}}.
    \end{aligned}
\end{equation}

With these definitions and the derived error subsystem dynamics \eqref{Eq:DG_error_dynamic} and \eqref{Eq:Line_error_dynamic}, it is easy to see that the closed-loop error dynamics of the DC MG can be modeled as a networked error system as shown in Fig. \ref{Fig.DissNetError}. In there, the interconnection relationship between the error subsystems, disturbance inputs and performance outputs is described by:
\begin{equation}
    \begin{bmatrix}
        \tilde{u} & \tilde{\bar{u}} & z_c
    \end{bmatrix}^\T = M \begin{bmatrix}
        \tilde{x} & \tilde{\bar{x}} & w_c
    \end{bmatrix}^\T,
\end{equation}
where the interconnection matrix $M$ takes the form:
\begin{equation}\label{Eq:NetErrSysMMat}
    M \triangleq 
        \bm{K & \bar{C} & E_c\\
        C & \0 & \bar{E}_c\\
        H_{c} & \bar{H}_{c} & \0
        }.
\end{equation}

\begin{figure}
    \centering
    \includegraphics[width=0.99\columnwidth]{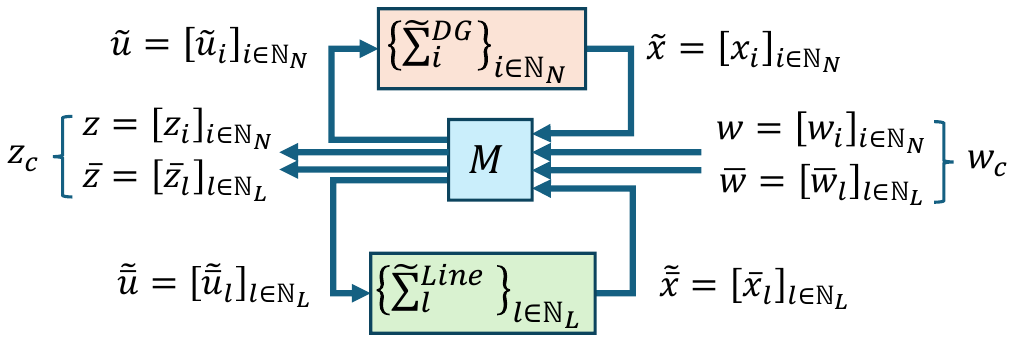}
    \caption{DC MG  error dynamics as a networked system with disturbance inputs and performance outputs.}
    \label{Fig.DissNetError}
\end{figure}

\section{Dissipativity-Based Co-Design Framework for Controllers and Communication Topology}\label{Passivity-based Control}

In this section, we first introduce the global control and topology co-design problem for the DC MG, using the networked error dynamics representation derived in the previous section (see Fig. \ref{Fig.DissNetError}). As this co-design problem exploits the dissipativity properties of the subsystem error dynamics, we next identify necessary conditions for subsystem dissipativity. Subsequently, we embed these necessary conditions in the local dissipative controller design problems implemented at the subsystems. Finally, the overall control design process is summarized.


\subsection{Error Subsystem Dissipativity Properties}\label{Sec:Diss_Property}
Consider the DG error subsystem $\tilde{\Sigma}_i^{DG},i\in\mathbb{N}_N$ \eqref{Eq:DG_error_dynamic} to be $X_i$-dissipative with
\begin{equation}\label{Eq:XEID_DG}
    X_i=\begin{bmatrix}
        X_i^{11} & X_i^{12} \\ X_i^{21} & X_i^{22}
    \end{bmatrix}\triangleq
    \begin{bmatrix}
        -\nu_i\mathbf{I} & \frac{1}{2}\mathbf{I} \\ \frac{1}{2}\mathbf{I} & -\rho_i\mathbf{I}
    \end{bmatrix},
\end{equation}
where $\rho_i$ and $\nu_i$ are the passivity indices of $\tilde{\Sigma}_i^{DG}$. In other words, consider $\tilde{\Sigma}_i^{DG}, i\in\mathbb{N}_N$ to be IF-OFP($\nu_i,\rho_i$). It is worth noting that the IF-OFP($\nu_i,\rho_i$) property assumed here will be enforced through local controller design in Sec. \ref{Sec:Local_Synth} (Th. \ref{Th:LocalControllerDesign}), where the passivity indices $\nu_i$ and $\rho_i$ are determined alongside the local controller gains $K_{i0}$.

Similarly, consider the line error subsystem $\tilde{\Sigma}_l^{Line},l\in\mathbb{N}_L$ \eqref{Eq:Line_error_dynamic} to be $\bar{X}_l$-dissipative with
\begin{equation}\label{Eq:XEID_Line}
    \bar{X}_l=\begin{bmatrix}
        \bar{X}_l^{11} & \bar{X}_l^{12} \\ \bar{X}_l^{21} & \bar{X}_l^{22}
    \end{bmatrix}\triangleq
    \begin{bmatrix}
        -\bar{\nu}_l\mathbf{I} & \frac{1}{2}\mathbf{I} \\ \frac{1}{2}\mathbf{I} & -\bar{\rho}_l\mathbf{I}
    \end{bmatrix},
\end{equation}
where $\bar{\rho}_l$ and $\bar{\nu}_l$ are the passivity indices of $\tilde{\Sigma}_l^{Line}$. Regarding these line-error subsystem passivity indices, we can state the following lemma.

\begin{lemma}\cite{najafirad2025dissipativity}\label{Lm:LineDissipativityStep}
For each line $\tilde{\Sigma}_l^{Line},\,  l\in\N_L$ \eqref{Eq:LineCompact}, its passivity indices $\bar{\nu}_l$,\,$\bar{\rho}_l$ assumed in (\ref{Eq:XEID_Line}) are such that the LMI problem: 
\begin{equation}\label{Eq:Lm:LineDissipativityStep1}
\begin{aligned}
\mbox{Find: }\ &\bar{P}_l, \bar{\nu}_l, \bar{\rho}_l\\
\mbox{Sub. to:}\ &\bar{P}_l > 0, \ 
\begin{bmatrix}
    \frac{2\bar{P}_lR_l}{L_l}-\bar{\rho}_l & -\frac{\bar{P}_l}{L_l}+\frac{1}{2}\\
    \star & -\bar{\nu}_l
\end{bmatrix}
\normalsize
\geq0, 
\end{aligned}
\end{equation}
is feasible. The maximum feasible values for $\bar{\nu}_l$ and $\bar{\rho}_l$ respectively are $\bar{\nu}_l^{\max}=0$ and $\bar{\rho}_l^{\max}=R_l$, when $\bar{P}_l =  \frac{L_l}{2}$. 
\end{lemma}





While we could identify the conditions required for the passivity indices of the line error dynamics \eqref{Eq:Line_error_dynamic}, achieving a similar feat for the DG error dynamics is not straightforward due to the nonlinear CPL and VSC input saturation terms involved (see \eqref{Eq:DG_error_dynamic}). This challenge is addressed in the following subsection.

\subsection{Sector-Bounded Characterization of Nonlinearities}
Fig. \ref{Fig:SectorBounds} illustrates example instances 
of the CPL nonlinearity $g_i(\tilde{V}_i)$ and the 
dead-zone nonlinearity $\phi_i(u_i)$, graphically 
highlighting their sector-bounded nature.

\begin{figure}[t]
    \centering
    \includegraphics[width=\columnwidth]{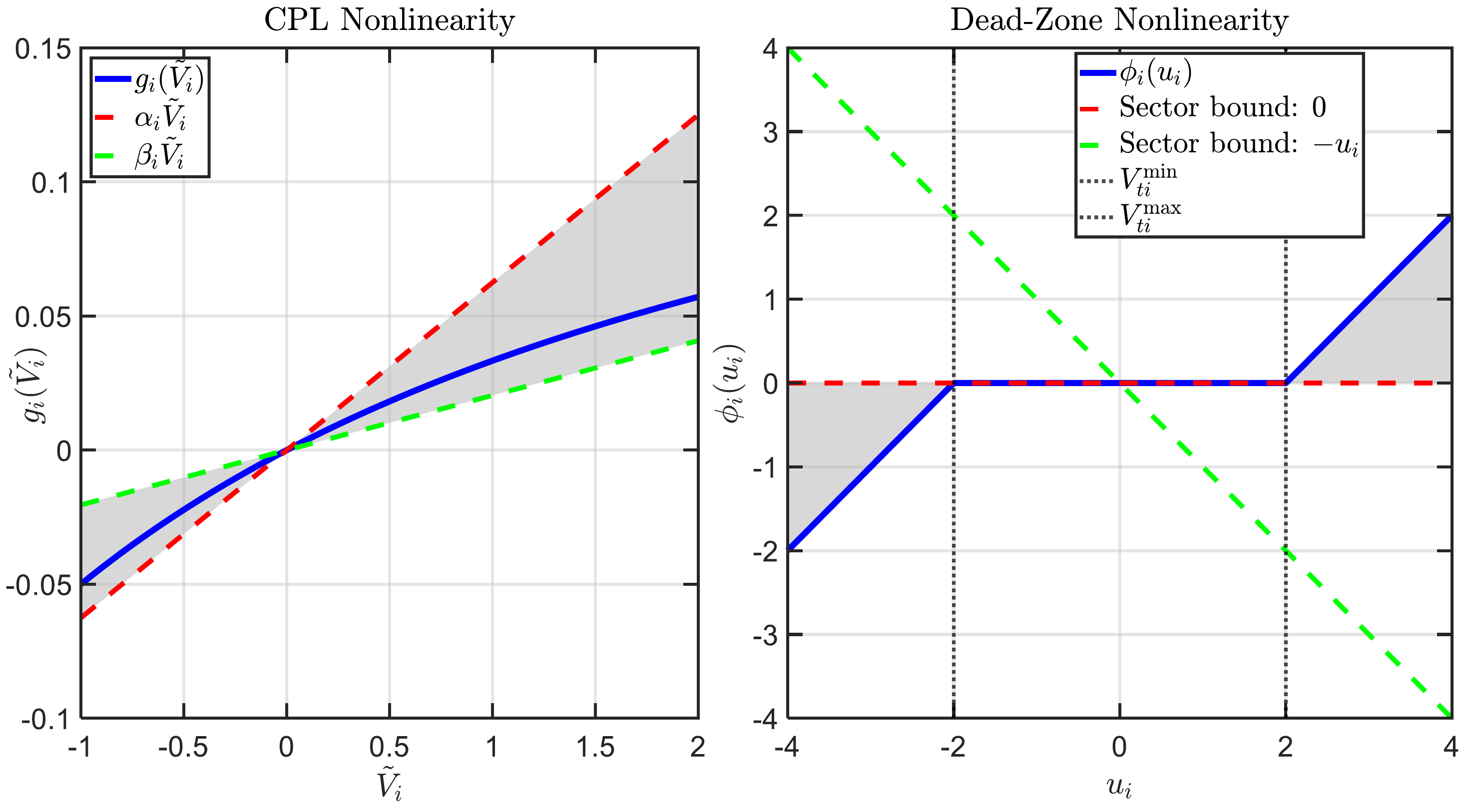}
    \caption{Sector-bounded characterizations of (left) 
    the CPL nonlinearity $g_i(\tilde{V}_i)$ bounded 
    between slopes $\alpha_i$ and $\beta_i$, and (right) 
    the dead-zone nonlinearity $\phi_i(u_i)$ bounded 
    within the sector $[0, -1]$.}
    \label{Fig:SectorBounds}
\end{figure}

The destabilizing negative-impedance characteristics of CPLs pose significant challenges to controller design, necessitating specialized mathematical treatment of their nonlinearities. In the following discussion, we present a systematic approach to incorporate CPL nonlinearities into our dissipativity-based control framework using sector-boundedness concepts. 

First, for notational convenience, we denote the first (and only non-zero) component of the CPL nonlinearity $g_i(\tilde{x}_i) \in \R^3$ \eqref{nonlinear_vector} as $g_i(\tilde{V}_i) \in \R$ (with a little abuse of notation), where 
\begin{equation} \label{nonlinear_scalar}
g_i(\tilde{V}_i) \triangleq \frac{P_{Li}}{C_{ti}}\left(\frac{1}{V_{ri}} - \frac{1}{\tilde{V}_i + V_{ri}}\right).
\end{equation}

The following lemma establishes the sector boundedness of this $g_i(\tilde{V}_i)$.

\begin{lemma}(Lm. 6, \cite{najafirad2025dissipativity})\label{Lm:sector_bounded}
The CPL nonlinearity $g_i(\tilde{V}_i)$ 
\eqref{nonlinear_scalar} satisfies the 
quadratic constraint:
\begin{equation}\label{Eq:Lm:sector_bounded}
\begin{bmatrix} \tilde{V}_i \\ g_i(\tilde{V}_i) 
\end{bmatrix}^\top 
\begin{bmatrix} 
-\alpha_i\beta_i & \frac{\alpha_i+\beta_i}{2} \\ 
\frac{\alpha_i+\beta_i}{2} & -1 
\end{bmatrix} 
\begin{bmatrix} \tilde{V}_i \\ g_i(\tilde{V}_i) 
\end{bmatrix} \geq 0,
\end{equation}
for all $\tilde{V}_i \in [\tilde{V}_i^{\min}, 
\tilde{V}_i^{\max}]\backslash\{0\}$, where 
$\alpha_i \triangleq \frac{P_{Li}}{V_{\max}^2}$, 
$\beta_i \triangleq \frac{P_{Li}}{V_{\min}^2}$, 
$\tilde{V}_i^{\min} \triangleq V_{\min} - V_{ri}$, 
and $\tilde{V}_i^{\max} \triangleq V_{\max} - V_{ri}$.
\end{lemma}

Next, we establish the sector-bounded characterization of the VSC input saturation nonlinearity. Recall that the dead-zone nonlinearity $\phi_i(u_i) = \text{sat}(u_i) - u_i$ captures the deviation of the saturated input from the unsaturated one. As illustrated in Fig. \ref{Fig:SectorBounds}, $\phi_i(u_i)$ and $u_i$ always have opposite signs or are both zero, meaning $\phi_i(u_i)$ is sector bounded within the sector $[0, -1]$ with respect to $u_i$. The following lemma formalizes this observation.

\begin{lemma}\label{Lm:Saturation_Sector_condition}
For the dead-zone nonlinearity $\phi_i(u_i)$ defined in 
\eqref{Eq:Saturation}, the following quadratic constraint 
holds for all $u_i \in \mathbb{R}$:
\begin{equation}
\begin{bmatrix} u_i \\ \phi_i(u_i) \end{bmatrix}^\top
\begin{bmatrix} 0 & -\frac{1}{2} \\ -\frac{1}{2} & -1 
\end{bmatrix}
\begin{bmatrix} u_i \\ \phi_i(u_i) \end{bmatrix} \geq 0.
\end{equation}
\end{lemma}
\begin{proof}
The constraint is equivalent to $\phi_i(u_i)[\phi_i(u_i) + 
u_i] \leq 0$. Since $\phi_i(u_i) + u_i = \text{sat}(u_i)$, 
this reduces to $\text{sat}(u_i) \cdot \phi_i(u_i) \leq 0$, 
which holds since $\text{sat}(u_i)$ and $\phi_i(u_i)$ always 
have opposite signs or are both zero (see \eqref{Eq:Saturation}).
\end{proof}

The IF-OFP property of each DG error subsystem $\tilde{\Sigma}_i^{DG}$ \eqref{Eq:DG_error_dynamic} can now be certified by exploiting Lm. \ref{Lm:sector_bounded} and Lm. \ref{Lm:Saturation_Sector_condition}.

\begin{assumption}\label{As:deltaBar}
At the local design stage, the distributed controller 
contribution $u_{iG}$ is treated as a bounded unknown 
satisfying $|u_{iG}| \leq \bar{\delta}_i$ for some 
$\bar{\delta}_i > 0$. The maximum tolerable bound 
$\bar{\delta}_i$ is recovered as part of the design 
via the change of variables $\tilde{\delta}_i \triangleq 
\sigma_i\bar{\delta}_i^2$, and once $\sigma_i$ and 
$\tilde{\delta}_i$ are obtained from the LMI, the 
maximum tolerable bound is recovered as 
$\bar{\delta}_i = \sqrt{\tilde{\delta}_i/\sigma_i}$.
\end{assumption}

\begin{theorem}\label{Th:Local_CPL}
Under As. \ref{As:deltaBar}, the DG error subsystem 
$\tilde{\Sigma}_i^{DG}: \tilde{\eta}_i \rightarrow 
\tilde{x}_i$, $i\in\mathbb{N}_N$ 
\eqref{Eq:DG_error_dynamic} can be made 
IF-OFP$(\nu_i, \rho_i)$ by designing the local 
controller gain $K_{i0}$ \eqref{Controller} via 
the LMI problem:
\begin{equation}\label{Eq:Th:Local_CPL}
\begin{aligned}
&\text{Find: }\ \tilde{K}_{i0},\ \tilde{P}_i,\ 
  \sigma_i,\ \tilde{\delta}_i,\ \nu_i,\ \tilde{\rho}_i,\\
&\text{Sub. to: }\ \tilde{P}_i > 0,\  
  \sigma_i \geq 0,\ \tilde{\rho}_i > 0,\ 
  \tilde{\delta}_i \geq 0,\\
&\begin{bmatrix}
{-}\nu_i\I & \tfrac{1}{2}\tilde{P}_i{-}\I & \0 
  & \0 & 0 & 0 & \0 \\
\star & \hat{\Delta}_i
  & \hat{\Phi}_i 
  & {-}B_i^{aw}{+}\tfrac{\mu_i}{2}
  \tilde{K}_{i0}^\top
  & 0 & 0 & \tilde{P}_i \\
\star & \star & \lambda_i
  & 0 & 0 & 0 & \0 \\
\star & \star & \star & \mu_i
  & {-}\tfrac{\mu_i}{2} & 0 & \0 \\
\star & \star & \star & \star & \sigma_i & 0 & \0 \\
\star & \star & \star & \star & \star 
  & \tilde{\delta}_i & \0 \\
\star & \star & \star & \star & \star 
  & \star & \tilde{\rho}_i\I
\end{bmatrix} \geq \0,
\end{aligned}
\end{equation}
where $\hat{\Delta}_i \triangleq 
{-}\mathcal{H}(\hat{A}_i) + \tilde{P}_ie_1e_1^\top
+ e_1e_1^\top\tilde{P}_i - \I$,
$\hat{A}_i \triangleq A_i\tilde{P}_i + B_i\tilde{K}_{i0}$,
$\hat{\Phi}_i \triangleq {-}e_1
{-}\tfrac{\alpha_i+\beta_i}{2\alpha_i\beta_i}\tilde{P}_ie_1$,
$e_1 \triangleq [1\ 0\ 0]^\top$,
$\lambda_i \triangleq \tfrac{1}{\alpha_i\beta_i}$ 
is uniquely determined by the CPL sector bounds 
$\alpha_i$ and $\beta_i$ from 
Lm. \ref{Lm:sector_bounded}, $\mu_i > 0$ is a 
prespecified scalar that linearizes the 
$\mu_i \times K_{i0}$ products, 
$K_{i0} \triangleq \tilde{K}_{i0}\tilde{P}_i^{-1}$, 
and $\rho_i \triangleq \tilde{\rho}_i^{-1}$. The 
maximum tolerable bound on $|u_{iG}|$ is recovered 
as $\bar{\delta}_i = \sqrt{\tilde{\delta}_i/\sigma_i}$.
\end{theorem}

\begin{proof}
Consider the storage function $\mathrm{V}_i(\tilde{x}_i) = 
\tilde{x}_i^\top P_i\tilde{x}_i$, $P_i > 0$. The 
IF-OFP$(\nu_i,\rho_i)$ property requires 
$\dot{\mathrm{V}}_i - s_i \leq 0$, where 
$s_i \triangleq -\nu_i\|\tilde{\eta}_i\|^2 - 
\rho_i\|\tilde{x}_i\|^2 + \tilde{\eta}_i^\top\tilde{x}_i$ 
is the IF-OFP supply rate. The scalar VSC command error 
decomposes as $u_i - u_{iE} = K_{i0}\tilde{x}_i + u_{iG}$,
where $u_{iG} \in \mathbb{R}$ is the distributed 
contribution with $u_{iGE} = 0$ at equilibrium. To apply 
Lm. \ref{Lm:Sprocedure}, the vector $\zeta_i$ is augmented 
with a fixed scalar $1$ to make the bound 
$\bar{\delta}_i^2 - u_{iG}^2 \geq 0$ a homogeneous 
quadratic form:
\begin{equation}\label{Eq:proof_zeta}
  \zeta_i \triangleq
  \bigl[
    \tilde{\eta}_i^\top\ \tilde{x}_i^\top\
    g_i\ \phi_i\ u_{iG}\ 1
  \bigr]^\top \in \mathbb{R}^{10}.
\end{equation}
Expanding $\dot{\mathrm{V}}_i$ along 
\eqref{Eq:DG_error_dynamic}, the condition 
$\dot{\mathrm{V}}_i - s_i \leq 0$ is equivalently 
expressed as $\zeta_i^\top\Pi_i\zeta_i \leq 0$, where 
$\hat{A}_i \triangleq A_i + B_iK_{i0}$ and $\Pi_i$ is 
as defined in \eqref{Eq:proof_matrices}, where the $u_{iG}$ 
and constant-$1$ rows/cols are zero since neither appears 
in $\dot{\mathrm{V}}_i - s_i$ directly.

From Lm. \ref{Lm:sector_bounded}, $g_i$ satisfies 
$\zeta_i^\top\Gamma_i^g\zeta_i \geq 0$ with $\Gamma_i^g$ 
as in \eqref{Eq:proof_matrices}. From 
Lm. \ref{Lm:Saturation_Sector_condition}, $\phi_i$ 
satisfies $-\phi_i(\phi_i + (u_i - u_{iE})) \geq 0$. 
Substituting $u_i - u_{iE} = K_{i0}\tilde{x}_i + u_{iG}$ 
and embedding the $2\times2$ sector matrix of 
Lm. \ref{Lm:Saturation_Sector_condition} into the 
$(\phi_i, u_i{-}u_{iE})$ block positions of $\zeta_i$, 
one obtains $\zeta_i^\top\Gamma_i^\phi\zeta_i \geq 0$ 
with $\Gamma_i^\phi$ as in \eqref{Eq:proof_matrices}, where 
block $(1,4) = \mathbf{0}$ since $\tilde{\eta}_i$ does 
not appear in $u_i - u_{iE}$ (note that $\tilde{\eta}_i$ 
is the interconnection input, not the control error), 
and block $(2,2) = \mathbf{0}_3$ since the $(1,1)$ entry 
of the sector matrix in 
Lm. \ref{Lm:Saturation_Sector_condition} is zero. Under 
As. \ref{As:deltaBar}, the bound $|u_{iG}| \leq 
\bar{\delta}_i$ gives a third quadratic constraint 
$\zeta_i^\top\Gamma_i^\delta\zeta_i \geq 0$ with 
$\Gamma_i^\delta$ as in \eqref{Eq:proof_matrices}, so that 
$\zeta_i^\top\Gamma_i^\delta\zeta_i = \bar{\delta}_i^2 - 
u_{iG}^2 \geq 0$.

Applying Lm. \ref{Lm:Sprocedure} to \eqref{Eq:proof_matrices}, a 
sufficient condition for \eqref{Eq:proof_matrices} to hold 
for all $\zeta_i$ satisfying the three constraints is 
that there exist $\lambda_i, \mu_i > 0$ and 
$\sigma_i \geq 0$ such that:
\begin{equation}\label{Eq:proof_Slemma}
\Pi_i - \lambda_i \Gamma_i^g - \mu_i\Gamma_i^\phi 
- \sigma_i\Gamma_i^\delta \geq 0.
\end{equation}
Note that $\lambda_i$, $\mu_i$, and $K_{i0}$ appear 
multiplied in \eqref{Eq:proof_Slemma}, creating bilinear 
terms. To linearize, we fix $\mu_i > 0$ as a prespecified 
scalar, which eliminates all $\mu_i \times K_{i0}$ 
products. Regarding $\lambda_i$, note that it multiplies 
$\Gamma_i^g$ whose entries involve $\alpha_i\beta_i$. 
Setting $\lambda_i \triangleq \frac{1}{\alpha_i\beta_i}$, 
which is uniquely determined by the CPL sector bounds 
from Lm. \ref{Lm:sector_bounded}, simplifies 
$\lambda_i\alpha_i\beta_i = 1$ and eliminates the 
$\lambda_i \times K_{i0}$ bilinearity without introducing 
any conservatism. With these choices, $\sigma_i \geq 0$ 
remains a free decision variable. Substituting \eqref{Eq:proof_matrices} and collecting terms, \eqref{Eq:proof_Slemma} becomes:
\begin{equation}\label{Eq:proof_6x6}
\begin{bmatrix}
{-}\nu_i\I_3 & \tfrac{1}{2}\I_3{-}P_i
  & \0 & \0 & 0 & 0 \\
\star & \Delta_i
  & \hat{\Phi}_i
  & {-}P_iB_i^{aw}{+}\tfrac{\mu_i}{2}K_{i0}^\top
  & 0 & 0 \\
\star & \star & \lambda_i
  & 0 & 0 & 0 \\
\star & \star & \star & \mu_i
  & {-}\tfrac{\mu_i}{2} & 0 \\
\star & \star & \star & \star & \sigma_i & 0 \\
\star & \star & \star & \star & \star
  & \sigma_i\bar{\delta}_i^2
\end{bmatrix}
\geq \mathbf{0},
\end{equation}
where $\Delta_i \triangleq {-}\rho_i\I_3 
{-}\mathcal{H}(P_i\hat{A}_i) {+}e_1e_1^\top$ and 
$\hat{\Phi}_i \triangleq {-}P_ie_1 
{-}\tfrac{\alpha_i+\beta_i}{2\alpha_i\beta_i} 
e_1 \in\mathbb{R}^{3\times1}$.

To handle the $\bar{\delta}_i^2$ term, we apply the 
change of variables $\tilde{\delta}_i \triangleq 
\sigma_i\bar{\delta}_i^2$, which allows the maximum 
tolerable bound $\bar{\delta}_i$ to be recovered as 
part of the design rather than being prespecified. 
Specifically, once $\sigma_i$ and $\tilde{\delta}_i$ 
are obtained from the LMI, the maximum tolerable bound 
is recovered as $\bar{\delta}_i = 
\sqrt{\tilde{\delta}_i/\sigma_i}$.

Applying the congruence transformation 
$\mathrm{diag}(\I_3,\tilde{P}_i,1,1,1,1)$ to 
\eqref{Eq:proof_6x6} with $\tilde{P}_i \triangleq 
P_i^{-1}$ and $\tilde{K}_{i0} \triangleq K_{i0}\tilde{P}_i$ 
linearizes $P_iK_{i0}$, yielding:
\begin{equation}\label{Eq:proof_6x6_transformed}
\begin{bmatrix}
{-}\nu_i\I_3 & \tfrac{1}{2}\tilde{P}_i{-}\I_3
  & \0 & \0 & 0 & 0 \\
\star & \bar{\Delta}_i
  & \hat{\Phi}_i
  & {-}B_i^{aw}{+}\tfrac{\mu_i}{2}
  \tilde{K}_{i0}^\top
  & 0 & 0 \\
\star & \star & \lambda_i 
  & 0 & 0 & 0 \\
\star & \star & \star & \mu_i
  & {-}\tfrac{\mu_i}{2} & 0 \\
\star & \star & \star & \star & \sigma_i & 0 \\
\star & \star & \star & \star & \star
  & \tilde{\delta}_i
\end{bmatrix}
\geq \mathbf{0},
\end{equation}
where $\hat{A}_i \triangleq A_i\tilde{P}_i + 
B_i\tilde{K}_{i0}$, $\hat{\Phi}_i \triangleq 
{-}e_1{-}\tfrac{\alpha_i+\beta_i}{2\alpha_i\beta_i} 
\tilde{P}_ie_1\in\mathbb{R}^{3\times1}$, and
\begin{equation}\label{Eq:proof_Deltabar}
\bar{\Delta}_i \triangleq
\underbrace{{-}\rho_i\tilde{P}_i^2}_{\text{Term A}}
- \mathcal{H}(\hat{A}_i)
+\underbrace{
\tilde{P}_ie_1e_1^\top\tilde{P}_i
}_{\text{Term B}}.
\end{equation}

\textit{Term A:} Applying Lm. \ref{Lm:Schur_comp} 
with $P = \tilde{\rho}_i\I_3$ and 
$Q = \bigl[\mathbf{0}_{3\times3}\ \tilde{P}_i\ 
\mathbf{0}_{3\times4}\bigr]$, the condition 
$R - Q^\top P^{-1}Q \geq 0$ is equivalent to augmenting 
\eqref{Eq:proof_6x6_transformed} with block 
$(2,7) = \tilde{P}_i$ and block $(7,7) = \tilde{\rho}_i\I_3$, 
completing the elimination of Term A.

\textit{Term B:} Applying Lm. \ref{Lm:InvSchur} with 
$P = \I_3$ and $Q = \tilde{P}_ie_1$:
\begin{equation*}
\tilde{P}_ie_1e_1^\top\tilde{P}_i
= Q^\top\I_3^{-1}Q \geq Q^\top + Q - \I_3
= \tilde{P}_ie_1e_1^\top + e_1e_1^\top\tilde{P}_i - \I_3.
\end{equation*}
Replacing Term B by this lower bound yields 
$\hat{\Delta}_i = {-}\mathcal{H}(\hat{A}_i) 
+\tilde{P}_ie_1e_1^\top +e_1e_1^\top\tilde{P}_i - \I_3$ 
as in \eqref{Eq:Th:Local_CPL}, completing the proof.

\begin{figure*}[b]
\hrulefill
\normalsize
\begin{equation}\label{Eq:proof_matrices}
\begin{aligned}
\Pi_i &\triangleq
\begin{bmatrix}
\nu_i\I_3 & P_i{-}\tfrac{1}{2}\I_3 & \0 & \0 & 0 & 0 \\
\star & \rho_i\I_3{+}\mathcal{H}(P_i\hat{A}_i) & P_ie_1 & P_iB_i^{aw} & 0 & 0 \\
\star & \star & 0 & 0 & 0 & 0 \\
\star & \star & \star & 0 & 0 & 0 \\
\star & \star & \star & \star & 0 & 0 \\
\star & \star & \star & \star & \star & 0
\end{bmatrix}, \quad
\Gamma_i^g \triangleq
\begin{bmatrix}
\0_3 & \0_3 & \0 & \0 & 0 & 0 \\
\star & {-}\alpha_i\beta_ie_1e_1^\top & \tfrac{\alpha_i+\beta_i}{2}e_1 & \0 & 0 & 0 \\
\star & \star & {-}1 & 0 & 0 & 0 \\
\star & \star & \star & 0 & 0 & 0 \\
\star & \star & \star & \star & 0 & 0 \\
\star & \star & \star & \star & \star & 0
\end{bmatrix}, \\[6pt]
\Gamma_i^\phi &\triangleq
\begin{bmatrix}
\0_3 & \0_3 & \0 & \0 & 0 & 0 \\
\star & \0_3 & \0 & {-}\tfrac{1}{2}K_{i0}^\top & 0 & 0 \\
\star & \star & 0 & 0 & 0 & 0 \\
\star & \star & \star & {-}1 & {-}\tfrac{1}{2} & 0 \\
\star & \star & \star & \star & 0 & 0 \\
\star & \star & \star & \star & \star & 0
\end{bmatrix}, \quad
\Gamma_i^\delta \triangleq
\begin{bmatrix}
\0_3 & \0_3 & 0 & 0 & 0 & 0 \\
\star & \0_3 & 0 & 0 & 0 & 0 \\
\star & \star & 0 & 0 & 0 & 0 \\
\star & \star & \star & 0 & 0 & 0 \\
\star & \star & \star & \star & {-}1 & 0 \\
\star & \star & \star & \star & \star & \bar{\delta}_i^2
\end{bmatrix}.
\end{aligned}
\end{equation}
\end{figure*}

\end{proof}

\subsection{Global Control and Topology Co-Design}

The local controllers \eqref{Controller} regulate the voltage at each DG while ensuring that closed-loop DG dynamics satisfy the required dissipativity properties established in Sec. \ref{Sec:Diss_Property}. Given these subsystem properties, we now synthesize the interconnection matrix $M$ \eqref{Eq:NetErrSysMMat} (see Fig. \ref{Fig.DissNetError}), particularly its block $K$, using Prop. \ref{synthesizeM}. Note that, by synthesizing $K=[K_{ij}]_{i,j\in\N_N}$, we can uniquely determine the consensus-based distributed global controller gains $\{k_{ij}^c:i,j\in\mathbb{N}_N\}$ \eqref{k_ij} (required in \eqref{ControllerG} to ensure the current sharing goal), along with the required communication topology $\mathcal{G}^c$. Note also that, when designing $K$ via Prop. \ref{synthesizeM}, we particularly enforce the closed-loop DC MG error dynamics to be $\textbf{Y}$-dissipative with $\textbf{Y} \triangleq \scriptsize \bm{\gamma^2\I & 0 \\ 0 & -\I}$ (see Rm. \ref{Rm:X-DissipativityVersions}) to prevent/bound the amplification of disturbances affecting the performance (voltage regulation and current sharing). The following theorem formulates this distributed global controller and communication topology co-design problem. 

\begin{theorem}\label{Th:CentralizedTopologyDesign}
The closed-loop networked error dynamics of the DC MG (see in Fig. \ref{Fig.DissNetError}) can be made finite-gain $L_2$-stable with an $L_2$-gain $\gamma$ (where $\Tilde{\gamma} \triangleq \gamma^2<\bar{\gamma}$ and $\bar{\gamma}$ is prespecified) from unknown disturbances $w_c(t)$ to performance output $z_c(t)$, by synthesizing the interconnection matrix block $M_{\tilde{u}x}=K$ (\ref{Eq:MMatrix}) via solving the LMI problem:
\begin{equation}
\label{Eq:Th:CentralizedTopologyDesign0}
\begin{aligned}
\min_{\substack{Q,\{p_i: i\in\N_N\},\\
\{\bar{p}_l: l\in\N_L\}, \tilde{\gamma}, S}} &\sum_{i,j\in\N_N} c_{ij} \Vert Q_{ij} \Vert_1 +c_1 \tilde{\gamma} + \alpha\text{tr}(S), \\
\mbox{Sub. to:}\ &p_i > 0,\ \forall i\in\N_N,\ 
\bar{p}_l > 0,\ \forall l\in\N_L,\\   
\ &\mbox{\eqref{globalcontrollertheorem}: } W + S > \0, \ S \geq \0,\ \text{tr}(S) \leq \eta,\\
\ &0 < \tilde{\gamma} < \bar{\gamma},\\
& Q_I P_n \textbf{1}_N = \0,
\end{aligned}
\end{equation}
as $K = (\textbf{X}_p^{11})^{-1} Q$ and $Q_I = [Q_{ij}^{2,2}]_{i,j\in\N_N}$, where 
$\textbf{X}^{12} \triangleq 
\diag([-\frac{1}{2\nu_i}\I]_{i\in\N_N})$, 
$\textbf{X}^{21} \triangleq (\textbf{X}^{12})^\T$,
$\Bar{\textbf{X}}^{12} \triangleq 
\diag([-\frac{1}{2\Bar{\nu}_l}\I]_{l\in\N_L})$,
$\Bar{\textbf{X}}^{21} \triangleq (\Bar{\textbf{X}}^{12})^\T$, 
$\textbf{X}_p^{11} \triangleq 
\diag([-p_i\nu_i\I]_{i\in\N_N})$, 
$\textbf{X}_p^{22} \triangleq 
\diag([-p_i\rho_i\I]_{i\in\N_N})$, 
$\Bar{\textbf{X}}_{\bar{p}}^{11} 
\triangleq \diag([-\bar{p}_l\bar{\nu}_l\I]_{l\in\N_L})$, 
$\Bar{\textbf{X}}_{\bar{p}}^{22} 
\triangleq \diag([-\bar{p}_l\bar{\rho}_l\I]_{l\in\N_L})$, and $\tilde{\Gamma} \triangleq \tilde{\gamma}\I$. 
The structure of $Q\triangleq[Q_{ij}]_{i,j\in\N_N}$ mirrors that of $K\triangleq[K_{ij}]_{i,j\in\N_N}$ (i.e., only the middle element is non-zero in each block $Q_{ij}$, see \eqref{k_ij}). 
The coefficients $c_1>0$ and $c_{ij}>0,\forall i,j\in\N_N$ are predefined cost coefficients corresponding to the $L_2$-gain (control cost) and communication links (communication cost), respectively. The matrix $S$ is a slack matrix included for numerical stability of the used LMI solver, where the slack coefficients $\alpha \geq 0$ and $\eta \geq 0$ respectively impose soft and hard constraints on $S$.
\end{theorem}


\begin{proof}
The proof follows by considering the closed-loop DC MG 
(shown in Fig. \ref{Fig.DissNetError}) as a networked 
system and applying the subsystem dissipativity properties 
assumed in \eqref{Eq:XEID_DG} and \eqref{Eq:XEID_Line} 
to the interconnection topology synthesis result given in 
Prop. \ref{synthesizeM}. 
The DG error subsystems are modeled as IF-OFP($\nu_i,\rho_i$) 
and line error subsystems as IF-OFP($\bar{\nu}_l,\bar{\rho}_l$), 
secured through local controller design and analysis in 
Th. \ref{Th:Local_CPL} and Lm. \ref{Lm:LineDissipativityStep}, 
respectively. Note that the IF-OFP($\nu_i,\rho_i$) certification 
in Th. \ref{Th:Local_CPL} accounts for both the CPL nonlinearity 
$g_i(\tilde{x}_i)$ and the VSC input saturation nonlinearity 
$\phi_i(u_i)$ simultaneously, via the S-procedure with 
prespecified multipliers $\lambda_i$ and $\mu_i$, respectively. 
The LMI problem \eqref{Eq:Th:CentralizedTopologyDesign0} is 
formulated to ensure the networked error system is 
$\mathbf{Y}$-dissipative, thereby ensuring finite-gain 
$L_2$-stability with gain $\gamma$ from disturbances $w_c$ 
to performance outputs $z_c$. The objective function in 
\eqref{Eq:Th:CentralizedTopologyDesign0} consists of three 
terms: communication cost 
($\sum_{i,j\in\mathbb{N}_N} c_{ij}\|Q_{ij}\|_1$), control 
cost $c_1\tilde{\gamma}$, and numerical stability term 
($\alpha\text{tr}(S)$). Minimizing this function while 
satisfying LMI constraints simultaneously optimizes the 
communication topology (by synthesizing 
$K = (\mathbf{X}_p^{11})^{-1}Q$) and robust stability 
(by minimizing $\tilde{\gamma}$) while ensuring 
$\gamma^2 < \bar{\gamma}$. The resulting controller and 
topology achieve voltage regulation and current sharing 
in the presence of ZIP loads and disturbances.    
\end{proof}

\begin{figure*}[!hb]
\vspace{-5mm}
\centering
\hrulefill
\begin{equation}\label{globalcontrollertheorem}
\scriptsize
	W \triangleq \bm{
		\textbf{X}_p^{11} & \0 & \0 & Q & \textbf{X}_p^{11}\Bar{C} &  \textbf{X}_p^{11}E_c \\
		\0 & \bar{\textbf{X}}_{\bar{p}}^{11} & \0 & \Bar{\textbf{X}}_{\Bar{p}}^{11}C & \0 & \bar{\textbf{X}}_{\bar{p}}^{11}\bar{E}_c \\
		\0 & \0 & \I & H_c & \bar{H}_c & \0 \\
		Q^\T & C^\T\Bar{\textbf{X}}_{\Bar{p}}^{11} & H_c^\T & -Q^\T\textbf{X}^{12}-\textbf{X}^{21}Q-\textbf{X}_p^{22} & -\textbf{X}^{21}\textbf{X}_{p}^{11}\bar{C}-C^\T\bar{\textbf{X}}_{\bar{p}}^{11}\bar{\textbf{X}}^{12} & -\textbf{X}^{21}\textbf{X}_p^{11}E_c \\
		\Bar{C}^\T\textbf{X}_p^{11} & \0 & \bar{H}_c^\T & -\Bar{C}^\T\textbf{X}_p^{11}\textbf{X}^{12}-\bar{\textbf{X}}^{21}\Bar{\textbf{X}}_{\Bar{p}}^{11}C & -\bar{\textbf{X}}_{\bar{p}}^{22} & -\bar{\textbf{X}}^{21}\Bar{\textbf{X}}_{\Bar{p}}^{11}\bar{E}_c \\ 
		E_c^\T\textbf{X}_p^{11} & \bar{E}_c^\T\Bar{\textbf{X}}_{\Bar{p}}^{11} & \0 & -E_c^\T\textbf{X}_p^{11}\textbf{X}^{12} & -\bar{E}_c^\T\Bar{\textbf{X}}_{\Bar{p}}^{11}\bar{\textbf{X}}^{12} & \tilde{\Gamma} \\
	}\normalsize 
 > \0 
\end{equation}
\end{figure*}

\begin{figure*}[!hb]
\vspace{-5mm}
\centering
\begin{equation}\label{Eq:Neccessary_condition}
\scriptsize
	\bm{
		-p_i\nu_i & 0 & 0 & 0 & -p_i\nu_i\bar{C}_{il} & -p_i\nu_i\\
		0 & -\bar{p}_l\bar{\nu}_l & 0 & -\bar{p}_l\bar{\nu}_lC_{il} & 0 & -\bar{p}_l\bar{\nu}_l \\
		0 & 0 & 1 & 1 & 1 & 0 \\
		0 & -C_{il}\bar{\nu}_l\bar{p}_l & 1 & p_i\rho_i & -\frac{1}{2}p_i\bar{C}_{il}-\frac{1}{2}C_{il}\bar{p}_l & -\frac{1}{2}p_i \\
		-\bar{C}_{il}\nu_ip_i & 0 & 1 & -\frac{1}{2}\bar{C}_{il}p_i-\frac{1}{2}\bar{p}_lC_{il} & \bar{p}_l\bar{\rho}_l & -\frac{1}{2}p_l \\ 
		-\nu_ip_i & -\bar{p}_l\bar{\nu}_l & 0 & -\frac{1}{2}p_i & -\frac{1}{2}\bar{p}_l & \tilde{\gamma}_i \\
	}\normalsize 
 > \0,\ \forall l\in \mathcal{E}_i, \forall i\in\N_N
\end{equation}
\end{figure*}

\begin{remark}
In the proposed co-design approach \eqref{Eq:Th:CentralizedTopologyDesign0}: (i) communication costs are minimized through sparse topology optimization, (ii) control performance is improved by minimizing the $L_2$-gain from disturbance inputs to performance outputs, and (iii) computational efficiency is not compromised through LMI formulation.    
\end{remark}

\subsection{Necessary Conditions on Subsystem Passivity Indices}

Based on the terms $\textbf{X}_p^{11}$, $\textbf{X}_p^{22}$, $\bar{\textbf{X}}_{\bar{p}}^{11}$, $\bar{\textbf{X}}_{\bar{p}}^{22}$, $\textbf{X}^{12}$, $\textbf{X}^{21}$, $\bar{\textbf{X}}^{12}$, and $\bar{\textbf{X}}^{21}$ appearing in \eqref{globalcontrollertheorem} included in the global co-design problem \eqref{Eq:Th:CentralizedTopologyDesign0}, it is clear that the feasibility and the effectiveness of the proposed global co-design technique (i.e., Th. \ref{Th:CentralizedTopologyDesign}) depends on the enforced passivity indices $\{(\nu_i,\rho_i):i\in\mathbb{N}_N\}$ \eqref{Eq:XEID_DG} and  $\{(\bar{\nu}_l,\bar{\rho}_l):l\in\mathbb{N}_L\}$ \eqref{Eq:XEID_Line} assumed for the DG error dynamics \eqref{Eq:DG_error_dynamic} and line error dynamics \eqref{Eq:Line_error_dynamic}, respectively.

However, using Th. \ref{Th:Local_CPL} for designing dissipativating local controllers in $\{u_{iL}:i\in\mathbb{N}_N\}$ \eqref{Controller}, we can obtain a specialized set of passivity indices for the DG error dynamics \eqref{Eq:DG_error_dynamic}. Similarly, using Lm. \ref{Lm:LineDissipativityStep} for dissipativity analyses, we can obtain a specialized set of passivity indices for the line error dynamics \eqref{Eq:Line_error_dynamic}. Hence, these local controller design and analysis processes have a great potential to impact the feasibility and effectiveness of the global co-design solution.

Therefore, when designing such local controllers (via Th. \ref{Th:Local_CPL}) and conducting such dissipativity analysis (via Lm. \ref{Lm:LineDissipativityStep}), one must also consider the specific conditions necessary for the feasibility and implications on the effectiveness of the eventual global co-design solution. The following lemma, inspired by \cite[Lm. 1]{WelikalaJ22022}, identifies local necessary conditions based on the global LMI conditions \eqref{Eq:Th:CentralizedTopologyDesign0} in the global co-design problem in Th. \ref{Th:CentralizedTopologyDesign}.

\begin{lemma}\label{Lm:CodesignConditions}
For the LMI conditions \eqref{Eq:Th:CentralizedTopologyDesign0} in Th. \ref{Th:CentralizedTopologyDesign} to hold, it is necessary that the passivity indices $\{\nu_i,\rho_i:i\in\mathbb{N}_N\}$ \eqref{Eq:XEID_DG} and  $\{\bar{\nu}_l,\bar{\rho}_l:l\in\mathbb{N}_L\}$ \eqref{Eq:XEID_Line} respectively enforced for the DG \eqref{Eq:DG_error_dynamic} and line \eqref{Eq:Line_error_dynamic} error dynamics are such that the LMI problem: 
\begin{equation}\label{Eq:Lm:CodesignConditions}
\begin{aligned}
\mbox{Find: }&\ \{(\nu_i,\rho_i,\tilde{\gamma}_i):i\in\N_N\},\{(\Bar{\nu}_l,\bar{\rho}_l):l\in\N_L\}\\
\mbox{Sub. to: }&\ 0 \leq \tilde{\gamma}_i \leq \bar{\gamma},\  \forall i\in\N_N,\ \eqref{Eq:Neccessary_condition},  
\end{aligned}
\end{equation} 
is feasible, where $p_i>0, \forall i\in\N_N$ and $\bar{p}_l>0, \forall l\in\N_L$ are some prespecified parameters. 
\end{lemma}

\begin{proof}
For the feasibility of the global co-design problem 
\eqref{Eq:Th:CentralizedTopologyDesign0}, $W$ given in 
\eqref{globalcontrollertheorem} must satisfy $W > 0$. 
Let $W = [W_{rs}]_{r,s\in\mathbb{N}_6}$ where each block 
$W_{rs}$ can be a block matrix of block dimensions 
$(N\times N)$, $(N\times L)$, or $(L\times L)$ depending 
on its location in $W$ (e.g., see blocks $W_{11}$, $W_{15}$, 
and $W_{22}$, respectively). Without loss of generality, 
let $W_{rs} \triangleq [W_{rs}^{jm}]_{j\in\bar{J}(r),
m\in\bar{M}(s)}$ where $\bar{J}(r),\bar{M}(s) \in \{N,L\}$. 
Inspired by \cite[Lm. 1]{WelikalaJ22022}, an equivalent 
condition for $W > 0$ is $\bar{W} \triangleq \text{BEW}(W) 
> 0$, where $\text{BEW}(W)$ is the ``block-elementwise'' 
form of $W$, created by combining appropriate inner-block 
elements of each block $W_{rs}$ into a block-block matrix: 
$\bar{W} = [[W_{rs}^{j,m}]_{r,s\in\mathbb{N}_6}]_{j\in
\mathbb{N}_{\bar{J}},m\in\mathbb{N}_{\bar{M}}}$. 
Considering only the diagonal blocks in $\bar{W}$ and the 
implication $\bar{W} > 0 \implies [[W_{rs}^{j,m}]_{r,s\in
\mathbb{N}_6}]_{j\in\mathbb{N}_{\bar{J}(r)},m\in\mathbb{N}_{
\bar{M}(r)}} > 0 \iff$ \eqref{Eq:Lm:CodesignConditions} 
(also recall $C_{il} \triangleq -C_{ti}\bar{C}_{il}^\top$ 
and $\bar{C}_{il} \triangleq -C_{ti}^{-1}[B_{il}\ 0\ 0]^\top$), 
we have \eqref{globalcontrollertheorem} $\implies$ 
\eqref{Eq:Lm:CodesignConditions}. In other words, 
\eqref{Eq:Lm:CodesignConditions} is a set of necessary 
conditions for the feasibility of 
\eqref{globalcontrollertheorem}.

Beyond supporting feasibility, the LMI problem 
\eqref{Eq:Lm:CodesignConditions}, through its inclusion 
of the constraint $0 \leq \tilde{\gamma}_i \leq \bar{\gamma}$, 
also improves the effectiveness of the global co-design 
\eqref{globalcontrollertheorem} by bounding the local 
$L_2$-gain at each DG subsystem.
\end{proof}

In conclusion, here we used the LMI problem \eqref{Eq:Th:CentralizedTopologyDesign0} to derive a set of necessary LMI conditions consolidated as a single LMI problem \eqref{Eq:Lm:CodesignConditions}. Ensuring the feasibility of this consolidated LMI problem \eqref{Eq:Lm:CodesignConditions} increases the feasibility and effectiveness of the LMI problem \eqref{Eq:Th:CentralizedTopologyDesign0} solution, i.e., of the global co-design. Finally, we also point out that the necessary conditions given in the LMI problem \eqref{Eq:Lm:CodesignConditions} are much stronger and complete than those given in our prior work \cite{ACCNajafi}.

\subsection{Local Controller Synthesis}\label{Sec:Local_Synth}

We conclude our proposed solution by providing the following theorem that integrates all the necessary LMI conditions for the global co-design of the DC MG (i.e., Th. \ref{Th:CentralizedTopologyDesign}), identified in Lm. \ref{Lm:CodesignConditions}, and use them simultaneously to design the local controllers for DG error dynamics and analyze local line error dynamics. In all, the following result removes the necessity of implementing/evaluating the LMI problems in Th. \ref{Th:Local_CPL}, Lm. \ref{Lm:LineDissipativityStep} and Lm. \ref{Lm:CodesignConditions} separately, and instead provides a unified LMI problem to lay the foundation required to execute the global control and topology co-design of the DC MG using the established Th. \ref{Th:CentralizedTopologyDesign}. 

\begin{theorem}\label{Th:LocalControllerDesign}
Under the predefined DG parameters 
\eqref{Eq:DGCompact}, line parameters 
\eqref{Eq:LineCompact}, and design parameters 
$\{p_i: i\in\mathbb{N}_N\}$, 
$\{\bar{p}_l: l\in\mathbb{N}_L\}$, the necessary 
conditions in 
\eqref{Eq:Th:CentralizedTopologyDesign0} hold 
if the local controller gains 
$\{K_{i0}, i\in\mathbb{N}_N\}$ \eqref{Controller} 
and DG and line passivity indices 
$\{\nu_i, \tilde{\rho}_i: i\in\mathbb{N}_N\}$ 
\eqref{Eq:XEID_DG} and 
$\{\bar{\nu}_l, \bar{\rho}_l: l\in\mathbb{N}_L\}$ 
\eqref{Eq:XEID_Line} are determined by solving 
the LMI problem:
\begin{equation}\label{Eq:Th:LocalControllerDesign}
\begin{aligned}
&\min \sum_{i=1}^{N}\alpha_{\gamma}\tilde{\gamma}_i,\\
&\text{Find: }\ 
  \{(\tilde{K}_{i0}, \tilde{P}_i, 
  \sigma_i, \tilde{\delta}_i, \nu_i, \tilde{\rho}_i, 
  \tilde{\gamma}_i): i\in\mathbb{N}_N\},\\ 
&\quad\quad\quad\ 
  \{(\bar{P}_l, \bar{\nu}_l, \bar{\rho}_l): 
  l\in\mathbb{N}_L\},\ 
  \{(\xi_{il}, s_1, s_2): 
  l\in\mathcal{E}_i, i\in\mathbb{N}_N\},\\
&\text{Sub. to: }\ 
  \tilde{P}_i > 0,\  
  \sigma_i \geq 0,\ \tilde{\rho}_i > 0,\ 
  \tilde{\delta}_i \geq 0,\
  \bar{P}_l > 0,\ 
  \eqref{Eq:Transformed_Necessary_condition},\\
&{\scriptsize
\begin{bmatrix}
{-}\nu_i\I & \tfrac{1}{2}\tilde{P}_i{-}\I & \0 
  & \0 & 0 & 0 & \0 \\
\star & \hat{\Delta}_i
  & \hat{\Phi}_i 
  & {-}B_i^{aw}{+}\tfrac{\mu_i}{2}
  \tilde{K}_{i0}^\top
  & 0 & 0 & \tilde{P}_i \\
\star & \star & \lambda_i
  & 0 & 0 & 0 & \0 \\
\star & \star & \star & \mu_i
  & {-}\tfrac{\mu_i}{2} & 0 & \0 \\
\star & \star & \star & \star & \sigma_i & 0 & \0 \\
\star & \star & \star & \star & \star 
  & \tilde{\delta}_i & \0 \\
\star & \star & \star & \star & \star 
  & \star & \tilde{\rho}_i\I
\end{bmatrix}} \geq \0,\\
&{\scriptsize\begin{bmatrix} 
  \frac{2\bar{P}_lR_l}{L_l}{-}\bar{\rho}_l 
  & {-}\frac{\bar{P}_l}{L_l}{+}\frac{1}{2}\\
  \star & -\bar{\nu}_l
\end{bmatrix}} \geq \0,\ \forall l\in\mathbb{N}_L,\\
&{\scriptsize\begin{bmatrix} 
  1 & \bar{\nu}_l & \tilde{\rho}_i \\
  \bar{\nu}_l & s_1 & \xi_{il} \\
  \tilde{\rho}_i & \xi_{il} & s_2
\end{bmatrix}} \geq \0,\ 
  \forall l\in\mathcal{E}_i,\ 
  \forall i\in\mathbb{N}_N,
\end{aligned}
\end{equation}
where
$\hat{A}_i \triangleq A_i\tilde{P}_i 
  + B_i\tilde{K}_{i0}$,
$\hat{\Delta}_i \triangleq 
  {-}\mathcal{H}(\hat{A}_i)
  + \tilde{P}_ie_1e_1^\top
  + e_1e_1^\top\tilde{P}_i - \I$,
$\hat{\Phi}_i \triangleq 
  {-}e_1
  {-}\tfrac{\alpha_i+\beta_i}{2\alpha_i\beta_i}
  \tilde{P}_ie_1$,
$e_1 \triangleq [1\ 0\ 0]^\top$,
$\lambda_i \triangleq \tfrac{1}{\alpha_i\beta_i}$ 
is uniquely determined by the CPL sector bounds 
$\alpha_i$ and $\beta_i$ from 
Lm. \ref{Lm:sector_bounded},
$\mu_i > 0$ is a prespecified scalar,
$\alpha_\gamma > 0$ is a weighting coefficient,
$K_{i0} \triangleq \tilde{K}_{i0}\tilde{P}_i^{-1}$,
$\rho_i \triangleq \tilde{\rho}_i^{-1}$,
and $\xi_{il}$ is an auxiliary variable 
enforcing $\xi_{il} \geq \bar{\nu}_l\tilde{\rho}_i$.
The maximum tolerable bound on $|u_{iG}|$ is 
recovered as 
$\bar{\delta}_i = \sqrt{\tilde{\delta}_i/\sigma_i}$.
\end{theorem}

\begin{proof}
The proof proceeds in four steps.

(i) The dynamic models of $\tilde{\Sigma}_i^{DG}$ 
and $\tilde{\Sigma}_l^{Line}$ described in 
\eqref{Eq:DGCompact} and \eqref{Eq:LineCompact} 
are considered, and the LMI-based synthesis and 
analysis techniques from Th. \ref{Th:Local_CPL} 
and Lm. \ref{Lm:LineDissipativityStep} are applied 
to enforce and identify the subsystem passivity 
indices assumed in \eqref{Eq:XEID_DG} and 
\eqref{Eq:XEID_Line}, respectively.

(ii) Th. \ref{Th:Local_CPL} is applied to obtain 
the local controller gains $K_{i0}$ and passivity 
indices $(\nu_i, \rho_i)$ for $\tilde{\Sigma}_i^{DG}$, 
leading to the LMI in 
\eqref{Eq:Th:LocalControllerDesign}. This LMI 
simultaneously certifies the IF-OFP$(\nu_i,\rho_i)$ 
property of $\tilde{\Sigma}_i^{DG}$ under both the 
CPL nonlinearity $g_i(\tilde{x}_i)$ and the VSC 
input saturation nonlinearity $\phi_i(u_i)$, 
including the anti-windup correction via $B_i^{aw}$, 
through the S-procedure with prespecified multiplier 
$\mu_i > 0$ and multiplier $\lambda_i \triangleq 
\tfrac{1}{\alpha_i\beta_i}$ uniquely determined by 
the CPL sector bounds, respectively. The change of 
variables $\tilde{\delta}_i \triangleq 
\sigma_i\bar{\delta}_i^2$ from 
As. \ref{As:deltaBar} is incorporated, replacing 
$\sigma_i\bar{\delta}_i^2$ in the LMI with the 
free variable $\tilde{\delta}_i$, from which the 
maximum tolerable bound is recovered as 
$\bar{\delta}_i = \sqrt{\tilde{\delta}_i/\sigma_i}$. 
Similarly, Lm. \ref{Lm:LineDissipativityStep} is 
applied to identify the passivity indices 
$(\bar{\nu}_l, \bar{\rho}_l)$ for 
$\tilde{\Sigma}_l^{Line}$, leading to the LMI in 
\eqref{Eq:Th:LocalControllerDesign}.

(iii) To handle the transformation from $\rho_i$ 
to $\tilde{\rho}_i$ where $\tilde{\rho}_i = 
\rho_i^{-1}$, a congruence transformation is 
applied to \eqref{Eq:Neccessary_condition} using 
$T = \mathrm{diag}(1, 1, 1, \rho_i^{-1}, 1, 1)$, 
upon which the term $p_i\rho_i$ at position $(4,4)$ 
of \eqref{Eq:Neccessary_condition} becomes 
$p_i\tilde{\rho}_i$, resulting in 
\eqref{Eq:Transformed_Necessary_condition}. 
The bilinear terms $\bar{\nu}_l\tilde{\rho}_i$ 
appearing in 
\eqref{Eq:Transformed_Necessary_condition} are 
handled by introducing auxiliary variables 
$\xi_{il}$ and enforcing $\xi_{il} \geq 
\bar{\nu}_l\tilde{\rho}_i$ via the Schur complement 
constraint:
\begin{equation*}
\begin{bmatrix} 
1 & \bar{\nu}_l & \tilde{\rho}_i \\
\bar{\nu}_l & s_1 & \xi_{il} \\
\tilde{\rho}_i & \xi_{il} & s_2
\end{bmatrix} \geq \mathbf{0},
\end{equation*}
where $s_1$ and $s_2$ are semidefinite 
optimization variables. This replaces each 
bilinear term $\bar{\nu}_l\tilde{\rho}_i$ by 
$\xi_{il}$ in 
\eqref{Eq:Transformed_Necessary_condition}, 
rendering the constraint LMI-compatible.

(iv) The necessary conditions identified in 
Lm. \ref{Lm:CodesignConditions} are imposed via 
\eqref{Eq:Transformed_Necessary_condition} to 
support the feasibility and effectiveness of the 
global co-design in 
Th. \ref{Th:CentralizedTopologyDesign}. 
This unified LMI problem yields a one-shot 
procedure that simultaneously designs the local 
controllers, determines the passivity indices, 
minimizes the local $L_2$-gain bounds 
$\tilde{\gamma}_i$, and recovers the maximum 
tolerable bound $\bar{\delta}_i$ on the distributed 
controller contribution, guaranteeing the 
feasibility of the global co-design problem.
\end{proof}

\begin{figure*}[!hb]
\vspace{-5mm}
\centering
\hrulefill
\begin{equation}\label{Eq:Transformed_Necessary_condition}
\scriptsize
	\bm{
		-p_i\nu_i & 0 & 0 & 0 & -p_i\nu_i\bar{C}_{il} & -p_i\nu_i\\
		0 & -\bar{p}_l\bar{\nu}_l & 0 & -\bar{p}_l\xi_{il}C_{il} & 0 & -\bar{p}_l\bar{\nu}_l \\
		0 & 0 & 1 & \tilde{\rho}_i & 1 & 0 \\
		0 & -C_{il}\xi_{il}\bar{p}_l & \tilde{\rho}_i & p_i\tilde{\rho}_i & -\frac{1}{2}p_i\bar{C}_{il}\tilde{\rho}_i-\frac{1}{2}C_{il}\bar{p}_l\tilde{\rho}_i & -\frac{1}{2}p_i\tilde{\rho}_i \\
		-\bar{C}_{il}\nu_ip_i & 0 & 1 & -\frac{1}{2}\bar{C}_{il}p_i\tilde{\rho}_i-\frac{1}{2}\bar{p}_lC_{il}\tilde{\rho}_i & \bar{p}_l\bar{\rho}_l & -\frac{1}{2}p_l \\ 
		-\nu_ip_i & -\bar{p}_l\bar{\nu}_l & 0 & -\frac{1}{2}p_i\tilde{\rho}_i & -\frac{1}{2}\bar{p}_l & \tilde{\gamma}_i \\
	}\normalsize 
 > \0,\ \forall l\in \mathcal{E}_i, \forall i\in\N_N
\end{equation}
\end{figure*}

\vspace{-2mm}
\section{Simulation Results}\label{Sec:Simulation}
\vspace{-1mm}
We evaluated the proposed framework on an islanded DC MG with 4 DGs and ZIP loads interconnected through 4 transmission lines, implemented in MATLAB/Simulink. The converters have a nominal voltage of 120 V with reference voltage $V_r = 48$ V, and VSC saturation limits $V_{ti}^{\min} = 0$ V and $V_{ti}^{\max} = 60$ V, $\forall i \in \mathbb{N}_N$, reflecting realistic hardware constraints. The reference voltages $V_r$ and current sharing coefficient $I_s$ are obtained as a feasible solution of Th. \ref{Th:Equilibrium}. The local and global LMI problems in Th. \ref{Th:LocalControllerDesign} and Th. \ref{Th:CentralizedTopologyDesign} were solved using YALMIP/MOSEK with design parameters $p_i = 0.1$, $\forall i \in \mathbb{N}_N$ and $\bar{p}_l = 0.01$, $\forall l \in \mathbb{N}_L$, confirming LMI feasibility and yielding well-defined passivity indices $\{\nu_i, \rho_i\}$ and controller gains $\{K_{i0}, K_{aw,i}\}$ for all DGs.

\textbf{\textit{Voltage Regulation and Current Sharing}}: Sequential load variations were applied to stress-test the framework: constant current loads at $t=1$ s, disconnection and reconnection of constant impedance loads at $t=4$ s and $t=7$ s, and addition of CPLs at $t=5$ s, which actively pushes VSC inputs toward their saturation limits. Fig. \ref{fig.voltagecurrent}(a) shows that all DG voltages successfully track the 48 V reference throughout all scenarios, with transient deviations rapidly attenuated, notably including the CPL-induced oscillations at $t=5$ s where saturation becomes active. Fig. \ref{fig.voltagecurrent}(b) confirms that proportional current sharing is maintained across all operating conditions, with per-unit currents converging to a common value despite load changes and saturation events. The anti-windup compensator prevents integrator windup during saturation-active intervals, ensuring smooth recovery without steady-state offset. Without the anti-windup compensator, integrator windup during saturation-active intervals would cause steady-state offset and potential instability.

\begin{figure}[t]
    \centering
    \includegraphics[width=0.98\columnwidth]{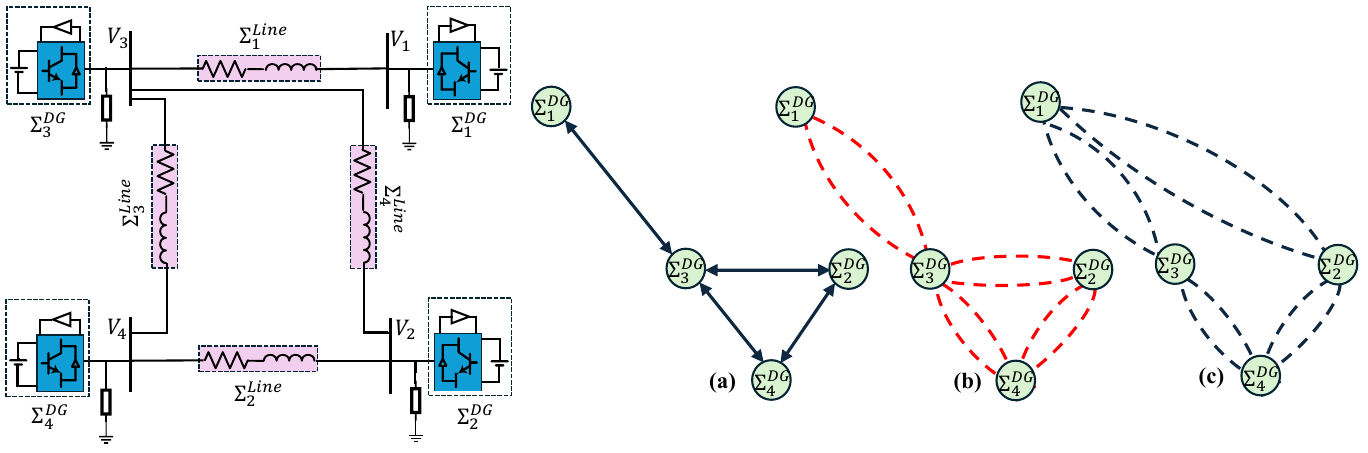}
    \vspace{-3mm}
    \caption{The islanded DC MG with 4 DGs and 4 lines: (a) physical topology, co-designed communication topology with anti-windup compensation under (b) hard graph and (c) soft graph constraints.}
    \label{fig.physicalcommunicationtopology}
\end{figure}

\begin{figure}[t]
    \centering
    \includegraphics[width=0.8\columnwidth]{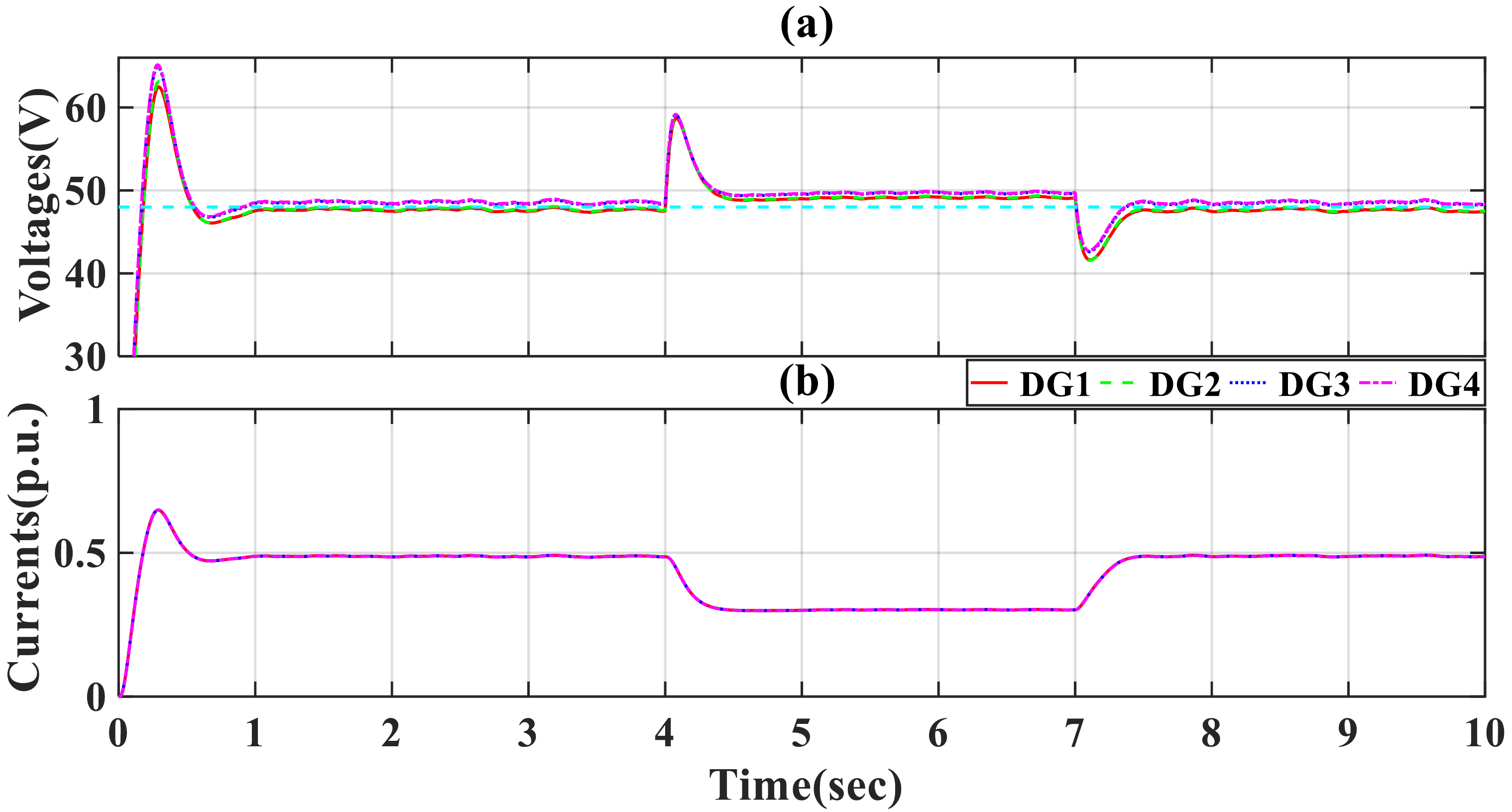}
    \vspace{-3mm}
    \caption{DG outputs: (a) voltages and (b) per-unit currents under the proposed dissipativity-based controller under ZIP loads and VSC-IS.}
    \label{fig.voltagecurrent}
\end{figure}

\textbf{\textit{Communication Topology Co-Design}}: Fig. \ref{fig.physicalcommunicationtopology} shows the physical topology and the two co-designed communication topologies obtained from Th. \ref{Th:CentralizedTopologyDesign}. Under the hard graph constraint, $\mathcal{G}^c$ is restricted to match $\mathcal{G}^p$, while the soft graph constraint allows topology optimization by penalizing deviations from $\mathcal{G}^p$. As shown in Fig. \ref{fig.physicalcommunicationtopology}(c), the soft constraint yields additional communication links that enhance robustness, particularly for $\Sigma_1^{DG}$ which benefits from additional coordination due to its physical isolation. Both topologies are obtained as a direct output of the global LMI \eqref{Eq:Th:CentralizedTopologyDesign0}, demonstrating that the co-design framework automatically trades off communication cost against closed-loop performance, a capability unavailable in sequential design methods where the topology is fixed prior to controller synthesis.

\vspace{-2mm}
\section{Conclusion}\label{Conclusion}
\vspace{-1mm}
This paper presents a dissipativity-based distributed control and communication topology co-design framework for DC MGs subject to ZIP loads and VSC-IS. A unified, non-iterative LMI-based framework is developed to simultaneously co-design local and distributed controllers, including an anti-windup compensator, and a communication topology, ensuring robust dissipativity from disturbance inputs to voltage regulation and current sharing performance outputs. The CPL nonlinearity and VSC-IS are each characterized via sector-boundedness, where the latter is handled through dead-zone decomposition, and both are simultaneously absorbed into the dissipativity analysis. The proposed approach eliminates droop coefficient tuning while maintaining proportional current sharing and can be efficiently evaluated using standard convex optimization tools. Future work will focus on extending the framework to AC and hybrid AC/DC microgrid configurations and experimental validation.

\section{Conclusion}\label{Conclusion}
This paper presents a dissipativity-based distributed droop-free control and communication topology co-design framework for DC microgrids subject to ZIP loads and VSC input saturation. By leveraging dissipativity theory and sector-boundedness concepts, a unified framework is developed to co-design local steady-state controllers, local feedback controllers, global distributed controllers, and a communication topology, so as to ensure robust dissipativity of the closed-loop DC MG from generic disturbance inputs to voltage regulation and current sharing performance outputs. The CPL nonlinearity and the VSC input saturation are each characterized via sector-boundedness, where the latter is handled through a dead-zone decomposition, and both are simultaneously absorbed into the dissipativity analysis through the S-procedure and Young's inequality. Unlike conventional droop-based methods, the proposed approach eliminates the need for precise droop coefficient tuning, enhancing voltage regulation accuracy while maintaining proportional current sharing. Moreover, the proposed framework is LMI-based and thus can be conveniently implemented and efficiently evaluated using existing standard convex optimization tools. Simulation results demonstrate the effectiveness and superior performance of the proposed framework compared to conventional droop control methods, particularly when handling CPLs and saturation-active operating conditions. Future work will focus on developing plug-and-play capabilities and extending the framework to AC microgrids and hybrid AC/DC microgrid configurations.

\bibliographystyle{IEEEtran}
\bibliography{References}

\end{document}